\documentclass{article}

\usepackage{amsmath}
\usepackage{amssymb}
\usepackage{amsthm}
\usepackage{tikz}
\usepackage{hyperref}
\usepackage[nofillcomment,boxed,linesnumbered]{algorithm2e}

\newcounter{global}
\theoremstyle{definition}

\theoremstyle{plain}
\newtheorem{theorem}[global]{Theorem}

\newtheoremstyle{note}{}{}{}{}{\itshape}{.}{.5em}{}
\theoremstyle{note}
\newtheorem{remark}{Remark}%
\newtheorem{example}{Example}%
\renewcommand\section{%
  \@startsection {section}{1}{\z@}%
  {-3.5ex \@plus -1ex \@minus -.2ex}%
  {2.3ex \@plus.2ex}%
  {\normalfont\large\bfseries}}
\tikzset{%
  vertex/.style = {circle, draw=black, fill=white,thick},
  edge/.style = {draw=black, thick}}

\def\assign{\ensuremath{\mathrel{:=}}}%
\def\Var{\ensuremath{\mathop{\mathrm{Var}}}}

\def\Z{\ensuremath{\mathbb{Z}}}

\def\Ax{\ensuremath{\mathop{(\mathrm{Ax})}}}
\def\Ref{\ensuremath{\mathop{(\mathrm{Ref})}}}
\def\Cut{\ensuremath{\mathop{(\mathrm{Cut})}}}
\def\Rwt{\ensuremath{\mathop{(\mathrm{Rwt})}}}
\def\Aug{\ensuremath{\mathop{(\mathrm{Aug})}}}
\def\Tra{\ensuremath{\mathop{(\mathrm{Tra})}}}
\def\Pro{\ensuremath{\mathop{(\mathrm{Pro})}}}
\def\logand{\ensuremath{\mathrel{\&}}}

\def\I{\Rightarrow}

\def\nmodels{\not\models}
\def\GBool{\ensuremath{\Gamma_{\!\mathbf{2}}}}
\def\reduces{\ensuremath{\rightharpoonup^*_\Gamma}}
\def\reducesd{\ensuremath{\rightharpoonup^*_\Delta}}
\def\rewrites{\ensuremath{\rightharpoonup_\Gamma}}

\def\derfrac#1#2{\ensuremath{\displaystyle%
    \cfrac{#1\rule[-1.05em]{0pt}{1em}}{#2\rule{0pt}{0.9em}}}}
\def\dderfrac#1#2{\ensuremath{\displaystyle%
    \cfrac{#1}{#2\rule{0pt}{0.9em}}}}
\def\xdderfrac#1#2{\ensuremath{\displaystyle%
    \cfrac{#1}{#2\rule{0pt}{0.9em}}}}
\def\lfrac#1#2#3{\derfrac{#1}{#2}\,{\scriptstyle #3}}
\def\llfrac#1#2#3{\dderfrac{#1}{#2}\,{\scriptstyle #3}}
\def\xlfrac#1#2#3{\xdderfrac{#1}{#2}\,{\scriptstyle #3}}
\def\AX#1{\ensuremath{\llfrac{}{#1}{\Ax}}}
\def\REF#1{\ensuremath{\llfrac{}{#1}{\Ref}}}

\def\TRA#1#2#3{\ensuremath{\lfrac{#1,\ #2}{#3}{\Tra}}}
\def\CUT#1#2#3{\ensuremath{\lfrac{#1,\ #2}{#3}{\Cut}}}
\def\RWT#1#2#3{\ensuremath{\lfrac{#1,\ #2}{#3}{\Rwt}}}
\def\xCUT#1#2#3{\ensuremath{\xlfrac{#1,\ #2}{#3}{\Cut}}}
\def\xRWT#1#2#3{\ensuremath{\xlfrac{#1,\ #2}{#3}{\Rwt}}}
\def\xAUG#1#2{\ensuremath{\xlfrac{#1}{#2}{\Aug}}}
\def\xPRO#1#2{\ensuremath{\xlfrac{#1}{#2}{\Pro}}}
\def\xTRA#1#2#3{\ensuremath{\xlfrac{#1,\ #2}{#3}{\Tra}}}

\def\atr#1{\ensuremath{\mathtt{#1}}}

\begin{document}

\title{Monoidal functional dependencies}

\date{\normalsize%
  Dept. Computer Science, Palacky University Olomouc}

\author{Vilem Vychodil\footnote{%
    e-mail: \texttt{vychodil@binghamton.edu},
    phone: +420 585 634 705,
    fax: +420 585 411 643}}

\maketitle

\begin{abstract}
  \noindent%
  We present a complete logic for reasoning with functional dependencies (FDs) with semantics defined over classes of commutative integral partially ordered monoids and complete residuated lattices. The dependencies allow us to express stronger relationships between attribute values than the ordinary FDs. In our setting, the dependencies not only express that certain values are determined by others but also express that similar values of attributes imply similar values of other attributes. We show complete axiomatization using a system of Armstrong-like rules, comment on related computational issues, and the relational vs. propositional semantics of the dependencies.
\end{abstract}

\section{Introduction}\label{sec:intro}
Rank-aware approaches in database systems~\cite{Il:ASOTQPTiRDS} represent
a popular alternative to traditional database systems which consider answers
to queries as sets of objects (e.g., sets of tuples of values in relational
systems). In contrast, rank-aware databases represent query results as sets of
objects together with scores. The role of scores is to express degrees to
which objects match queries. The primary interpretation of scores is
comparative---higher scores represent better matches. Most of the existing
rank-aware approaches focus on issues related to efficient query evaluation
in order to show only $k$ best matches (results with $k$ best scores) to
a query. The various approaches differ in how they achieve this goal,
see~\cite{Il:ASOTQPTiRDS} for a survey.

In this paper, we study a logic for a new type of dependencies that appear
in particular rank-aware database systems. Namely, we are interested in
approaches which (i) evaluate atomic queries as sets of
objects with scores and (ii) express scores in results of composed
conjunctive queries by applying monotone aggregation functions
to the scores obtained from evaluating the subqueries. In fact,
these are particular types of queries which appear in the influential
paper of R.~Fagin~\cite{Fa98:CFIfMS} (cf. also~\cite{Fagin:Oaafm}) dealing
with monotone query evaluation. In the sense of~\cite{Fa98:CFIfMS}, answers
to a query like
\begin{align}
  \atr{LOCATION} = \text{\texttt{"Byron St"}}
  \logand
  \atr{AREA} = \text{\texttt{2,400}}
  \logand
  \atr{PRICE} = \text{\texttt{\$800,000}}
  \label{eqn:q}
\end{align}
which represents a request for houses in (or near) \texttt{Byron St},
with floor size of (or about) \texttt{2,400} square feet,
and sold at \texttt{\$800,000} (or similar price) are determined
by evaluating all three subqueries for each object (a house for sale)
in the database, obtaining three scores. Then, the three scores for each
object are aggregated by a monotone function to get the score for the object
in the result of the conjunctive query~\eqref{eqn:q}. In the same way as
users may be interested only in the best few answers to queries
like~\eqref{eqn:q}, we may argue that maintainers of the database may
be interested in imposing constraints which take the scores (i.e., degrees
of matches) into account. For instance,
\begin{align}
  (\atr{LOCATION} \logand \atr{AREA}) \I \atr{PRICE}
  \label{eqn:fmlex}
\end{align}
is syntactically an ordinary FD but we can give it a new semantics from the
point of view of the scores and the aggregation function:
A relation $\mathbf{r}$ satisfies~\eqref{eqn:fmlex} if for any two tuples
in $\mathbf{r}$, similar values of locations and similar values of areas
imply similar prices. For any two tuples $r_1$ and $r_2$,
we may formalize the condition as
\begin{align}
  & (r_1(\atr{LOCATION}) \approx r_2(\atr{LOCATION}))
  \otimes 
  (r_1(\atr{AREA}) \approx r_2(\atr{AREA}))
  \leq \notag \\
  & r_1(\atr{PRICE}) \approx r_2(\atr{PRICE}),
  \label{eqn:fmlex*}
\end{align}
where $\otimes$ is the monotone aggregation function which interprets
the conjunction denoted above as $\logand$ and $\approx$ assigns to any
two values $d_1$ and $d_2$ of the same type a score which is the result
of atomic query $d_1 = d_2$.
Let us note that $\leq$ in~\eqref{eqn:fmlex*} is used to interpret
the material implication $\I$ in~\eqref{eqn:fmlex}.
This reflects the fact that in the classical
propositional logic, a formula $\varphi \I \psi$ is true under evaluation $e$
if{}f the truth value of $\varphi$ under $e$ is less than or
equal to the truth value of $\psi$ under $e$. In~\eqref{eqn:fmlex*},
we have just applied this principle to scores instead of
the logical $0$ and $1$ (which may be seen as two borderline scores).
In general, \eqref{eqn:fmlex*} represents a stronger relationship than
that represented by the ordinary FD semantics: the condition can be violated
if two tuples have close values of locations and area but considerably larger
difference between prices. In this sense, the illustrative
formula~\eqref{eqn:fmlex} can be seen as a constraint in a rank-aware
database, ensuring that houses of similar properties (locations and area)
should be offered for similar prices, thus avoiding unwanted situations of
underpriced or overpriced offers.

The approach in~\cite{Fa98:CFIfMS} of efficient query execution relies on
aggregation functions defined on the real unit interval which are
monotone and strict. Typically, triangular norms are used to this
extent but~\cite{Fa98:CFIfMS} is even more general (it has been exploited in
various approaches which are not truth functional, cf.~\cite{DaReSu:Pdditd}).
We consider more general structures than those defined on the real unit
interval. In order to interpret~\eqref{eqn:fmlex} as in~\eqref{eqn:fmlex*},
it suffices to have a set $L$ of scores which can be compared by a partial
order relation $\leq$ on $L$ and with $1 \in L$ being the highest score
(representing a full match). Moreover, we need an aggregation function
$\otimes$ which should be associative and commutative (because the bracketing
and the order of the conjunctive subqueries should not matter) with $1$ being
its neutral element.
In addition, $\otimes$ should be monotone w.r.t. $\leq$
which ensures that better matches of subqueries yield higher scores in
the result. These conditions imply the condition of strictness
from~\cite{Fa98:CFIfMS}. Altogether, we base our considerations on structures
of scores which are in fact partially ordered Abelian monoids from which comes
the term ``monoidal FDs'' (shortly, an MFD).

In this paper, we primarily focus on logic for reasoning with formulas
like~\eqref{eqn:fmlex} which is different from the logic for reasoning
with FDs, because we interpret the formulas over general monoidal structures
and not Boolean algebras. For instance, $\otimes$ is not idempotent
in general (on $L = [0,1]$ with its natural ordering,
the only idempotent $\otimes$ is the minimum).
In practice this means that the number of occurrences of propositional
variables (i.e., the names of attributes in database terminology) in
formulas matters and it enables us to express weaker or stronger
relationships between attributes. For illustration,
\begin{align}
  (\atr{LOCATION} \logand \atr{AREA} \logand \atr{AREA}) \I \atr{PRICE}
  \label{eqn:weaker}
\end{align}
is a formula which prescribes a weaker constraint than~\eqref{eqn:fmlex}
because the truth value of its antecendent (under a given evaluation) is
in general lower than (or equal to) the truth value of the
antecedent of~\eqref{eqn:fmlex} (under the same evaluation).
Thus, if~\eqref{eqn:fmlex} is satisfied then so is~\eqref{eqn:weaker}
but not \emph{vice versa} in general. Analogously,
\begin{align}
  (\atr{LOCATION} \logand \atr{AREA}) \I (\atr{PRICE} \logand \atr{PRICE})
  \label{eqn:stronger}
\end{align}
prescribes a stronger constraint than~\eqref{eqn:fmlex}. Indeed,
the truth value of its consequent is in general lower than or equal to
the truth value of the consequent of~\eqref{eqn:fmlex}, i.e.,
if \eqref{eqn:stronger} is satisfied then so is~\eqref{eqn:fmlex}
but not \emph{vice versa} in general. So, the very presence of
non-idempotent conjunctions allows us to put more/less emphasis
on similarity-based constraints.
Let us also note that~\cite{Fa98:CFIfMS} considers general non-idempotent
functions interpreting $\logand$ as well. As a result, writing
$\atr{AREA} = \text{\texttt{2,400}}$ twice in a query like \eqref{eqn:q}
changes the meaning of the query by putting more emphasis on the area
being close to the specified value and the query may produce a different
result. So, accepting non-idempotent interpretations of $\logand$ in
rank-aware approaches to query evaluation as in~\cite{Fa98:CFIfMS} or
data dependencies as we present here should not be surprising,
cf. also~\cite{HaPa:Adfl} for an informal discussion on topics
related to non-idempotent conjunctions.

Using a different technique than is usual in the ordinary case,
we establish a complete axiomatization of our logic which resembles
the well known Armstrong rules~\cite{Arm:Dsdbr}.
This makes our approach different from
other approaches which tackle similar issues but focus almost exclusively
on idempotent conjunctions; we present more details on the relationship to
other approaches in Section~\ref{sec:compl}.
A survey and a comparison of relevant approaches in this direction can
be found in~\cite{BeVy:Codd}.

Our approach is not limited only to the database (relational) semantics.
In fact, we start with a propositional semantics over monoidal structures
and later prove that there is a relational semantics which yields the same
notion of semantic entailment (and thus has the same axiomatization). This
is analogous to~\cite{Fagin} (cf. also~\cite{DeCa,SaDePaFa:Ebrddfpl})
which shows that the logic
of the classic FDs is in fact a particular propositional fragment. In this
sense, the logic of MFDs we describe in the paper is a particular
propositional fragment of H\"ohle's monoidal logic~\cite{Ho:ML}.

In much the same way as the classic functional dependencies, MFDs serve two
basic purposes. First, they can be used as formulas prescribing constraints.
Second, they can be used as formulas derived from database instances,
describing dependencies that hold in data. While the first role may be
expected and is traditionally studied in databases, the second one seems
to be of equal importance and is more related to data analysis and data
mining. Our paper offers a sound and complete logic system which can be
used as a formal basis for both types of problems.

The paper is organized as follows. In Section~\ref{sec:prelim},
we present preliminaries from partially ordered structures we utilize in
the paper. In Section~\ref{sec:log}, we present the syntax and semantics of
our logic and in Section~\ref{sec:compl}, we prove its completeness.
In Section~\ref{sec:comput}, we deal with related computational issues.
In Section~\ref{sec:propdb}, we discuss the relationship between
two possible interpretations of formulas used in this paper.
In Section~\ref{sec:relwork}, we present a survey of the most relevant
related work. Finally, in Section~\ref{sec:concl} we present conclusion
and open problems.

\section{Preliminaries}\label{sec:prelim}
We assume that readers are familiar with the basic notions of partially ordered
sets (posets) and lattices. A \emph{partially ordered monoid}
(shortly, a pomonoid) is
a structure $\mathbf{L} = \langle L,\leq,\otimes,1\rangle$ where
$\langle L,\otimes,1\rangle$ is a monoid (i.e., a semigroup with neutral
element $1$), and $\leq$ is a partial order on $L$ so that $\otimes$
is monotone w.r.t. $\leq$: If $a \leq b$, then $a \otimes c \leq b \otimes c$
and $c \otimes a \leq c \otimes b$. Furthermore, if $1$ is the greatest
element of $L$ w.r.t. $\leq$, then $\mathbf{L}$ is
called \emph{integral pomonoid}. In the paper, we work mostly with
integral \emph{commutative} pomonoids (i.e., $\otimes$ is
in addition commutative). Given $\mathbf{L}$, $a \in L$ and non-negative
integer $n$, we define the \emph{$n$th power} $a^n$ of $a$ by putting
$a^0 = 1$ and $a^{n+1} = a \otimes a^n$ for each natural $n$.

Related structures which appear in various substructural logics
are residuated lattices~\cite{Di38,DiWa}.
An integral commutative residuated lattice
(shortly, a \emph{residuated lattice}) is
a structure $\mathbf{L} = \langle L,\wedge,\vee,\otimes,\rightarrow,0,1\rangle$
such that $\langle L,\wedge,\vee,0,1\rangle$ is a bounded lattice,
$\langle L,\leq,\otimes,1\rangle$ is an integral commutative pomonoid
($\leq$ is the lattice order from $\mathbf{L}$, i.e., $a \leq b$ if{}f
$a = a \wedge b$), and $\rightarrow$ satisfies, for all $a,b,c \in L$,
$a \otimes b \leq c$ if{}f $a \leq b \rightarrow c$ (so-called adjointness
property). The operations $\otimes$ (called a multiplication) and
$\rightarrow$ (called a residuum) serve as general interpretations of
logical connectives ``conjunction'' and ``implication''.
In addition, $\mathbf{L}$ is called \emph{complete} if
$\langle L,\wedge,\vee,0,1\rangle$ is a complete lattice.
The adjointness
property ensures that $\otimes$ and $\rightarrow$ are general enough
and still have desirable properties---the important role of the adjointness
condition in logics has been discovered by J.\,A. Goguen~\cite{Gog:Lic}.
We mention here one property that is relevant to this paper: As a consequence
of the adjointness, $a \leq b$ if{}f $a \rightarrow b = 1$ (easy to see).
The class of residuated lattices is definable by identities and therefore it
forms a variety. The variety has interesting subvarieties, including
a subvariety which is term-equivalent to the variety of Boolean algebras.
In particular,
$\mathbf{L} = \langle L,\wedge,\vee,\otimes,\rightarrow,0,1\rangle$,
where $L=\{0,1\}$, $\otimes = \wedge$, and $\wedge,\vee,\rightarrow$
are truth functions of the classic conjunction, disjunction, and implication,
respectively, is the structure of truth degrees of the classic
propositional logic~\cite{Me87}. Most widely known multiple-valued (fuzzy)
logics based on subclasses of residuated lattices are BL~\cite{Haj:MFL}
and MTL~\cite{EsGo:MTL} which are the logics of all continuous and
left-continuous triangular norms~\cite{KMP:TN}, respectively.
More details on residuated structures and their role in logic and relational
systems may be found in~\cite{Bel:FRS,Bir:LT,GaJiKoOn:RL,Wec:UAfCS},
cf. also the recent edited book~\cite{CiHaNo1}.

\section{Monoidal Functional Dependencies: \\ Syntax and Semantics}\label{sec:log}
In this section, we formalize the rules, present their interpretation,
and introduce an inference system for deriving rules from sets of other rules.
Note that in this section, we present a propositional semantics of the rules
which generalizes the interpretation of analogous formulas in the classic
propositional logic. In Section~\ref{sec:propdb}, we introduce a relational
semantics which is equivalent to the propositional one.

From the logical point of view, our rules are implications between two
formulas containing conjunctions of propositional variables which can occur
in the formulas multiple times. This is in contrast
to the classic FDs where the number of occurrences does not matter---what
matters is whether a propositional variable (in database systems called
an attribute) is present in the formula or not. Consequently, classic FDs
are often presented as implications between \emph{sets} of propositional
variables which simplifies many considerations on FDs including their
axiomatization and computation of closures.

In our setting, we cannot make such simplification because conjunctions are
interpreted by aggregation functions which are not idempotent in general.
On the other hand, the functions are still commutative and associative.
Therefore, we can disregard the order in which propositional variables
appear in formulas and the bracketing. We may therefore introduce the
following notation:
If $\mathrm{Var}$ is a denumerable set of propositional variables,
we consider maps of the form
\begin{align}
  A\!: \Var \to \Z
  \label{eqn:A}
\end{align}
satisfying both of the following conditions:
\begin{enumerate}\parskip=0pt
\item
  $A(p) \geq 0$ for all $p \in \Var$,
\item
  $\{p \in \Var;\, A(p) > 0\}$ is finite.
\end{enumerate}
The maps can be seen as finite multi-subsets of \Var{} and we use them to
formalize antecedents and consequents of if-then formulas.
In particular, we denote by $\top$ a~map of the form \eqref{eqn:A}
such that $\top(p) = 0$ for all $p \in \Var$, i.e., $\top$ can be seen
as an empty multi-subset of \Var.

Now, by a \emph{monoidal functional dependency} (shortly, an MFD), we mean
any expression of the form $A \I B$, where both $A,B$ are of the
form~\eqref{eqn:A}. Clearly, such rules can be seen as shorthands for
formulas like~\eqref{eqn:fmlex}; in this case,
$A(\atr{LOCATION})=1$, $A(\atr{AREA})=1$, $B(\atr{PRICE})=1$,
and $A$ abd $B$ take the value $0$ at each other atribute.

MFDs are interpreted with respect to evaluations of propositional variables
which assign to each propositional variable an element from the support of
an integral commutative pomonoid. The situation is fully analogous to
evaluations in the classic case which assign to propositional
variables two logical values $0$ and $1$.

Formally,
let $\mathbf{L} = \langle L,\leq,\otimes,1\rangle$ be an integral
commutative pomonoid.
An \emph{$\mathbf{L}$-evaluation} (shortly, an \emph{evaluation} if $\mathbf{L}$
is clear from context) is any map $e\!: \Var \to L$. That is, each evaluation
$e$ assigns to each propositional variable $p \in \Var$ a degree $e(p) \in L$.
The degree $e(p)$ is interpreted as the degree to which $p$ is satisfied under
the evaluation $e$. Each evaluation can be uniquely extended to all
maps~\eqref{eqn:A}:
For $A$ of the form~\eqref{eqn:A}, we define $e(A) \in L$ as follows
\begin{align}
  e(A) &= e(p_1)^{A(p_1)} \otimes \cdots \otimes e(p_n)^{A(p_n)},
  \label{eqn:eA}
\end{align}
where $\{p \in \Var;\, A(p) > 0\} \subseteq \{p_1,\ldots,p_n\}$. Recall
from the preliminaries that the powers which appear in~\eqref{eqn:eA}
are considered with respect to the monoidal operation $\otimes$ in
$\mathbf{L}$. Thus, $e(p_1)^{A(p_1)}$ means $e(p_1)$ $\otimes$-multiplied
by itself $A(p_1)$-times. Also note that by definition, we get $a^0 = 1$.
Thus, the value of~\eqref{eqn:eA} depends only on variables $p \in \Var$
such that $A(p) > 0$. As a special case, we have $e(\top) = 1$ because
$1$ is neutral with respect to $\otimes$.

For $A \I B$ and $\mathbf{L}$-evaluation $e$, we say that $A \I B$ is
\emph{satisfied under $e$}, written $e \models A \I B$ whenever
$e(A) \leq e(B)$, where $\leq$ is the partial order in $\mathbf{L}$.
Furthermore, $A \I B$
is called \emph{an $\mathbf{L}$-tautology} if it is satisfied in
any $\mathbf{L}$-evaluation (with $\mathbf{L}$ fixed);
$A \I B$ is called \emph{trivial} if it is $\mathbf{L}$-tautology
for any~$\mathbf{L}$.

We now introduce semantic entailment of MFDs in terms of models.
Suppose that $\Var$ is fixed. A set $\Gamma$ of MFDs is called
a \emph{theory} (over $\Var$).
Each $\mathbf{L}$-evaluation $e$ such that $e \models A \I B$ for all
$A \I B \in \Gamma$ is called an \emph{$\mathbf{L}$-model of $\Gamma$.}
An MFD $A \I B$ is \emph{semantically entailed by $\Gamma$,}
written $\Gamma \models A \I B$, if $e \models A \I B$
for any $\mathbf{L}$-model $e$ of $\Gamma$ with $\mathbf{L}$ being any
integral commutative pomonoid.

\begin{remark}
  Our notion of the semantic entailment is not dependent on a particular
  choice of $\mathbf{L}$ because $\Gamma \models A \I B$ if{}f
  for any $\mathbf{L}$ and any $\mathbf{L}$-evaluation $e$,
  we get $e(A) \leq e(B)$. Also note that our logic is consistent in
  that each theory has an $\mathbf{L}$-model $e$ for any $\mathbf{L}$
  (take $e(p)=1$ for all $p \in \Var$).
\end{remark}

In the paper we show that $\models$ can be characterized syntactically.
In case of MFDs, the need for a syntactic
characterization of $\models$ seems to be more important than in the case
of classic FDs because the semantic entailment, by its definition, involves
checking $e \models A \I B$ over all $\mathbf{L}$-models where $\mathbf{L}$
ranges over all integral commutative pomonoids which is a proper class of
algebras. In contrast, the entailment of FDs can be checked by efficient
linear-time algorithms~\cite{BeBe:Cprttdonfrs}.

In the inference rules introduced below,
we use the following notation. For maps $A,B$ of
the form~\eqref{eqn:A}, we define a map 
$AB\!: \Var \to \Z$ by
\begin{align}
  (AB)(p) = A(p) + B(p)
  \label{eqn:AB}
\end{align}
for any $p \in \Var$.
In addition, we put $A^0 = \top$ and $A^{n+1} = AA^{n}$ for any
natural $n$ and call $A^n$ the $n$th power of $A$.
Obviously, $\{p \in \Var;\, (AB)(p) > 0\}$ is a finite
set and therefore $AB$ as well as $A^n$ are maps of the form~\eqref{eqn:A}.
Our use of maps like~\eqref{eqn:AB} is analogous to the set-theoretic
union which is used in inference rules for the classic FDs.

In our logic, we consider the following two inference rules:
\bgroup%
\addtolength{\leftmargini}{1.8em}%
\begin{itemize}%
\item[\Ax]
  infer $AB \I B$,
\item[\Cut]
  from $A \I B$ and $BC \I D$ infer $AC \I D$,
\end{itemize}%
\egroup%
\noindent%
where $A,B,C,D$ are arbitrary maps~\eqref{eqn:A}. As usual, a sequence
$\varphi_1,\ldots,\varphi_n$ of MFDs is called a \emph{proof} of $\varphi_n$
by a theory $\Gamma$ if each $\varphi_i$ is in $\Gamma$ or is derived
from $\varphi_1,\ldots,\varphi_{i-1}$ by \Ax{} or \Cut. Notice that \Ax{}
is in fact a nullary rule (an axiom scheme) which derives $AB \I B$ from
no input formulas. In this sense, \Cut{} is the only (non-trivial) inference
rule in our system which infers new formulas from existing ones.
In database literature~\cite{Mai:TRD}, the classic counterpart of \Cut{}
is often called \emph{pseudotransitivity.}
An MFD $A \I B$ is called \emph{provable} by $\Gamma$, written
$\Gamma \vdash A \I B$, if there is a proof of $A \I B$ by $\Gamma$.

\begin{remark}\label{rem:inf}
  (a)
  For convenience, we may write \Ax{} and \Cut{} in a ``fraction notation''
  like
  \begin{align*}
    &\AX{AB \I B}, &
    &\xCUT{A \I B}{BC \I D}{AC \I D}.
  \end{align*}
  and write proofs by $\Gamma$ as trees with leaves corresponding to
  formulas in $\Gamma$ and internal nodes given by instances of~\Ax{}
  and \Cut.

  (b)
  Let us note that complete systems of inference rules for the classic FDs
  (Armstrong systems~\cite{Arm:Dsdbr}) are usually presented in less compact
  way using \Ax{} (sometimes called
  the axiom of reflexivity) and the following rules
  \begin{align*}
    &\xTRA{A \I B}{B \I C}{A \I C},
    &\xAUG{A \I B}{AC \I BC} &
  \end{align*}
  instead of \Cut{}. This can also be done in our case.
  Indeed, \Tra{} is a particular case of \Cut{} for $C = \top$ and \Aug{}
  results by \Cut{} from $A \I B$ and $BC \I BC$ which is an instance of \Ax.
  Conversely, in order to show that \Cut{} is derivable from \Tra{} and \Aug,
  observe that
  \begin{align*}
    \TRA{\xAUG{A \I B}{AC \I BC}}{BC \I D}{AC \I D}.
  \end{align*}
  Let us note that even if \Ax{} and \Cut{} as well as the other rules look
  syntactically similar to their classic counterparts, the rules do not
  operate on implications between \emph{sets} of attributes and, therefore,
  represent different rules. In general,
  \Ax{} and \Cut{} in our logic are \emph{weaker rules} than their
  set-theoretic counterparts.
  For instance, our system admits the following weaker form
  of additivity:
  \vspace*{-\baselineskip}%
  \begin{align*}
    \CUT{A \I C}{\xCUT{A \I B}{\AX{BC \I BC}\rule[-1em]{0pt}{1em}}{AC \I BC}}{AA \I BC},
  \end{align*}
  i.e., a rule which from $A \I B$ and $A \I C$ infers $AA \I BC$ but
  in general, the ordinary-style additivity~\cite{Mai:TRD} which
  infers $A \I BC$ from $A \I B$ and $A \I C$ is not sound and thus
  not derivable in our logic as we shall see in the next section.

  (c)
  In Section~\ref{sec:comput}, we utilize an alternative system of
  inference rules which resemble the classic B-axioms~\cite[page 52]{Mai:TRD}.
  Namely, we consider the following rules of reflexivity, rewriting,
  and projectivity:
  \begin{align*}
    &\REF{A \I A}, &
    &\xRWT{A \I BC}{C \I D}{A \I BD}, &
    &\xPRO{A \I BC}{A \I B}.
  \end{align*}
  Note that the original B-axioms use $\Ref$ and $\Pro$ together with
  the rule of accumulation which infers $A \I BCD$ from
  $A \I BC$ and $C \I DE$. As in the case of additivity, we can show
  that accumulation in this form is not derivable in our system. Our rule
  $\Rwt$ may be seen as a weaker form of the accumulation and its name
  reflects the fact that $C$ appearing in $A \I BC$ is replaced by $D$ and
  is not kept in the derived formula $A \I BD$. The inference rules
  $\Ref$, $\Rwt$, and $\Pro$ are equivalent to $\Ax$ and $\Cut$.
  Indeed, $\Ref$ is an instance of $\Ax$, $\Rwt$ is obtained by $\Aug$ and
  $\Tra$, and $\Pro$ is obtained by $\Ax$ and $\Tra$. Conversely, $\Ax$ is
  obtained by $\Ref$ and $\Pro$, and $\Cut$ is obtained by $\Ref$
  and $\Rwt$ applied twice:
  \vspace*{-.8\baselineskip}%
  \begin{align*}
    \RWT{\RWT{\REF{AC \I AC}}{A \I B}{AC \I BC}}{BC \I D}{AC \I D},
  \end{align*}
  showing that $\Ax$ and $\Cut$ are equivalent to $\Ref$, $\Rwt$, and $\Pro$.
\end{remark}

\section{Completeness}\label{sec:compl}
We start investigating soundness and completeness of the inference system
with respect to the semantic entailment introduced in the previous
section. First, note that directly from~\eqref{eqn:eA} and~\eqref{eqn:AB},
\begin{align}
  e(AB) = e(A) \otimes e(B)
  \label{eqn:eAB}
\end{align}
for any maps $A,B$ like~\eqref{eqn:A}. As a consequence, $e(A^n) = e(A)^n$.
Our first observation identifies trivial MFDs and instances of \Ax. In its
proof, we use a special notation for writing particular maps of the
form \eqref{eqn:A}. Namely, for $p \in \Var$, we consider $\alpha_p$
such that
\begin{align}
  \alpha_p(q) &=
  \left\{
    \begin{array}{@{\,}l@{\quad}l@{}}
      1, & \text{if } p = q, \\
      0, & \text{otherwise,}
    \end{array}
  \right.
  \label{eqn:Ap}
\end{align}
for all $q \in \Var$. Note that for any $\mathbf{L}$-evaluation $e$ and
any $p \in \Var$, we have $e(p) = e(\alpha_p)$. Therefore, if there is no
danger of confusing propositional variables and maps of the form \eqref{eqn:A},
we write just $p,q,\ldots$ to denote $\alpha_p,\alpha_q,\ldots$, and the like.
This allows us to write, e.g., $ppq$ as an abbreviation for
$\alpha_p\alpha_p\alpha_q$ and we have
$e(ppq) = e(p) \otimes e(p) \otimes e(q) = e(\alpha_p\alpha_p\alpha_q)$
according to~\eqref{eqn:eA} and \eqref{eqn:AB}.

\begin{theorem}\label{th:trivial}
  $A \I B$ is trivial if{}f $A \I B$ is an instance of \Ax.
\end{theorem}
\begin{proof}
  Consider an $\mathbf{L}$-evaluation $e$. We get
  $e(AB) = e(A) \otimes e(B) \leq 1 \otimes e(B) = e(B)$.
  Indeed, the first equality comes from~\eqref{eqn:eAB};
  the next inequality is a consequence of the monotony of $\leq$
  and the fact that $1$ is the greatest element of $\mathbf{L}$;
  and the last equality follows from the fact that $1$ is the neutral
  element of $\otimes$. Hence, $e(AB) \leq e(B)$ yields $e \models AB \I B$,
  i.e., instances of \Ax{} are trivial.

  Conversely, we find
  an $\mathbf{L}$-model which satisfies only the trivial MFDs.
  Let $\mathbf{L} = \langle L,\leqslant,\cdot,\top\rangle$ be a structure
  where $L$ is the set of all maps~\eqref{eqn:A} for fixed $\Var$,
  $\cdot$ is a binary operation defined by $A \cdot B = AB$ as
  in~\eqref{eqn:AB}, and $A \leqslant B$ if{}f $B(p) \leq A(p)$
  for all $p \in \Var$. Clearly, $\mathbf{L}$ is an integral commutative
  pomonoid. In addition, consider $\mathbf{L}$-evaluation $e$ such that
  $e(p) = \alpha_p$ with $\alpha_p$ defined as in~\eqref{eqn:Ap}.
  It is easily seen that $e$ extends to all maps like~\eqref{eqn:A}
  so that $e(A) = A$ for any $A \in L$.
  Now, if $A \I B$ is not an
  instance of \Ax{}, then there is $p$ such that $A(p) < B(p)$
  and thus $e(A) = A \nleqslant B = e(B)$, showing $e \nmodels A \I B$.
\end{proof}

\begin{theorem}[soundness]\label{th:snd}%
  If $\Gamma \vdash A \I B$ then $\Gamma \models A \I B$.
\end{theorem}
\begin{proof}
  Assume that $e \models A \I B$ and $e \models BC \I D$ for
  $\mathbf{L}$-evaluation $e$.
  It means $e(A) \leq e(B)$ and $e(BC) = e(B) \otimes e(C) \leq e(D)$.
  Thus, utilizing the monotony of $\otimes$ and the transitivity of $\leq$,
  $e(A) \otimes e(C) \leq e(D)$, meaning that $e(AC) \leq e(D)$ which
  proves $e \models AC \I D$.
  The rest follows by induction on the length of a proof,
  utilizing Theorem~\ref{th:trivial}.
\end{proof}

\begin{remark}\label{rem:nsound}
  Take $\mathbf{L} = \langle [0,1],\leq,\otimes,1\rangle$ which is the
  commutative monoid of reals restricted to the interval $[0,1]$ with
  $\leq$ and $\otimes$ being the genuine ordering and multiplication
  of reals, respectively.
  Take $\mathbf{L}$-evaluation
  $e$ such that $e(p) = 0.5$ and $e(q) = e(r) = 0.6$.
  Thus, $e(p) = 0.5 \leq 0.6 = e(q)$ and analogously
  for $p$ and~$r$.
  On the other hand, $e(p) \nleq 0.36 = 0.6 \otimes 0.6 =
  e(qr)$. Therefore,
  $e \models p \I q$,
  $e \models p \I r$, and
  $e \nmodels p \I qr$,
  showing that $\{p \I q,p \I r\}
  \nmodels p \I qr$.
  Using Theorem~\ref{th:snd}, $p \I qr$ is not
  provable by $\{p \I q,p \I r\}$
  which shows that the classic rule of additivity is not derivable
  in our system, cf. Remark~\ref{rem:inf}\,(b). In a similar way,
  one can show that $p \I qrs$ is not provable by $\{p \I qr, r \I st\}$
  and thus the classic rule of accumulation is not derivable
  in our system (consider $e$ such that $e(q) = e(t) = 1$
  and $e(p) = e(r) = e(s) = 0.6$), cf. Remark~\ref{rem:inf}\,(c).
\end{remark}

The classic proof of completeness of inference rules for the classic FDs
involves closures of sets of attributes and exploits the property that
for each $A \subseteq R$ (where $R$ is a finite set of attributes)
the set $\{B \subseteq R;\, \Gamma \vdash A \I B\}$ has a greatest
element with respect to $\subseteq$. This property no longer holds
in our case (hint: see the previous Remark).
Nevertheless, we are able to prove strong completeness
(for general infinite $\Gamma$) by a technique which involves the construction
of a model from equivalence classes based on provability by $\Gamma$.
The procedure in the proof of the following theorem can be seen as
construction of the Lindenbaum algebra~\cite{RaSi} for a logic with
a restricted set of formulas which only take form of implications
between conjunctions of propositional variables.

\begin{theorem}[completeness]\label{th:cmp}%
  $\Gamma \vdash A \I B$ if{}f $\Gamma \models A \I B$.
\end{theorem}
\begin{proof}
  The only-if part follows by Theorem~\ref{th:snd}. We prove the
  if-part indirectly. Assuming that $\Gamma \nvdash A \I B$, we find
  an $\mathbf{L}$-model $e$ of $\Gamma$ such that $e(A) \nleq e(B)$.

  Let $L$ denote the set of all maps of the form~\eqref{eqn:A}
  for a fixed denumerable $\Var$ such that all propositional variables
  which occur in all formulas in $\Gamma$ are contained in $\Var$.
  Furthermore, consider the commutative monoid
  $\langle L,\cdot,\top\rangle$ as in the proof of Theorem~\ref{th:trivial}
  (the partial order is not considered at this point).
  The monoid is further used to express
  the desired model of $\Gamma$ in which $A \I B$ is not satisfied.

  Define binary relation $\equiv_\Gamma$ on $L$ as follows:
  $E \equiv_\Gamma F$ if{}f $\Gamma \vdash E \I F$ and $\Gamma \vdash F \I E$.
  We claim that $\equiv_\Gamma$ is a congruence relation
  on $\langle L,\cdot,\top\rangle$. In order to see that, we must check that
  $\equiv_\Gamma$ is equivalence and is compatible with $\cdot$ from 
  $\langle L,\cdot,\top\rangle$.
  Obviously, $\equiv_\Gamma$ is reflexive because of \Ax{} and is symetric
  by its definition. Since \Tra{} is a special case of \Cut, we can also
  conclude that $\equiv_\Gamma$ is transitive, i.e., it is an equivalence
  relation. Now, assume that $E \equiv_\Gamma F$ and $G \equiv_\Gamma H$.
  We have
  \vspace*{-.5\baselineskip}%
  \begin{align}
    \CUT{E \I F}{\CUT{G \I H}{\AX{FH \I FH}}{FG \I FH}}{EG \I FH},
    \label{eqn:ACBD}
  \end{align}
  i.e., from $\Gamma \vdash E \I F$ and $\Gamma \vdash G \I H$, it follows
  that $\Gamma \vdash EG \I FH$. Dually, 
  $\Gamma \vdash F \I E$ and $\Gamma \vdash H \I G$ yield
  $\Gamma \vdash FH \I EG$, showing $EG \equiv_\Gamma FH$.

  Therefore, $\equiv_\Gamma$ is a congruence relation and we may
  consider the quotient algebra $\mathbf{L}/\Gamma$ of $\mathbf{L}$
  modulo $\equiv_\Gamma$.
  In a more detail, $\mathbf{L}/\Gamma =
  \langle L/\Gamma,\circ,[\top]_\Gamma\rangle$,
  where $\mathbf{L}/\Gamma$ consists of all the equivalence classes
  $[{\cdots}]_\Gamma$ of $\equiv_\Gamma$, $[E]_\Gamma \circ [F]_\Gamma =
  [E \cdot F]_\Gamma = [EF]_\Gamma$,
  and $[\top]_\Gamma$ is the equivalence class containing~$\top$.
  Since commutative monoids form a variety~\cite{Wec:UAfCS},
  $\mathbf{L}/\Gamma = \langle L/\Gamma,\circ,[\top]_\Gamma\rangle$ 
  is also a commutative monoid. In addition, it can be equipped with
  a relation $\leqslant_\Gamma$ as follows: We put
  $[E]_\Gamma \leqslant_\Gamma [F]_\Gamma$
  whenever $\Gamma \vdash E \I F$.
  Again, using \Ax{} and \Cut,
  it follows that $\leqslant_\Gamma$ is a partial order
  on $L/\Gamma$ and its definition does not depend on the choice of elements
  from the equivalence classes---this is easy to see, we omit details.
  Moreover, $\circ$ is monotone with respect to $\leqslant_\Gamma$. Indeed,
  if $[E]_\Gamma \leqslant_\Gamma [F]_\Gamma$ and
  $[G]_\Gamma \leqslant_\Gamma [H]_\Gamma$,
  then from $\Gamma \vdash E \I F$ and $\Gamma \vdash G \I H$,
  we get $\Gamma \vdash EG \I FH$ as in~\eqref{eqn:ACBD}, showing
  $[EG]_\Gamma \leqslant_\Gamma [FH]_\Gamma$. In addition, we can see
  that $[E]_\Gamma \leqslant_\Gamma [\top]_\Gamma$ on account of
  $\Gamma \vdash E \I \top$ which is a trivial consequence of \Ax,
  i.e., $[\top]_\Gamma$ is the greatest element
  with respect to $\leqslant_\Gamma$.
  Altogether, $\mathbf{L}/\Gamma =
  \langle L/\Gamma,\leqslant_\Gamma,\circ,[\top]_\Gamma\rangle$
  is a commutative integral pomonoid.

  Take $\mathbf{L}/\Gamma$-evaluation $e$ such that
  $e(p) = [\alpha_p]_\Gamma$,
  where $\alpha_p\!: \Var \to \Z$ is defined as in~\eqref{eqn:Ap}.
  Observe how $e$ extends to all maps $E$ of the from~\eqref{eqn:A}.
  According to~\eqref{eqn:eA},
  \begin{align}
    e(E) &= e(p_1)^{E(p_1)} \circ \cdots \circ e(p_n)^{E(p_n)}
    = [\alpha_{p_1}]^{E(p_1)}_\Gamma \circ \cdots \circ
    [\alpha_{p_n}]^{E(p_n)}_\Gamma \notag \\
    &= \bigl[\alpha_{p_1}^{E(p_1)}\bigr]_\Gamma \circ \cdots \circ
    \bigl[\alpha_{p_n}^{E(p_n)}\bigr]_\Gamma =
    \bigl[\alpha_{p_1}^{E(p_1)} \cdots \alpha_{p_n}^{E(p_n)}\bigr]_\Gamma =
    [E]_\Gamma,
    \label{eqn:esharp}
  \end{align}
  where $\{p \in \Var;\, E(p) > 0\} \subseteq \{p_1,\ldots,p_n\}$.
  We now show that such $e$ is an $\mathbf{L}/\Gamma$-model of $\Gamma$.
  Take any $E \I F \in \Gamma$. Trivially, $\Gamma \vdash E \I F$ and
  thus \eqref{eqn:esharp} yields
  $e(E) = [E]_\Gamma \leqslant_\Gamma [F]_\Gamma = e(F)$,
  showing $e \models E \I F$.
  Since we have assumed $\Gamma \nvdash A \I B$,
  we get $e(A) = [A]_\Gamma \nleqslant_\Gamma [B]_\Gamma = e(B)$
  which shows that $e \nmodels A \I B$ and therefore $\Gamma \nmodels A \I B$.
\end{proof}

As a further demonstration of properties of $\vdash$ which is weaker than
the provability of classic FDs, we show the following variant of
a deduction-like theorem~\cite{Me87}:

\begin{theorem}[local deduction theorem]\label{th:ded}%
  Let $\Gamma$ be a theory. Then, the following are equivalent:
  \begin{itemize}
  \item[\rm (i)]
    there is natural $n$ such that $\Gamma \vdash A^n \I B$,
  \item[\rm (ii)]
    $\Gamma \cup \{\top \I A\} \vdash \top \I B$.
  \end{itemize}
\end{theorem}
\begin{proof}
  Assume that $\Gamma \vdash A^n \I B$ for some natural $n$.
  Since $\vdash$ is monotone, we get
  $\Gamma \cup \{\top \I A\} \vdash A^n \I B$. Applying \Cut,
  we get $\Gamma \cup \{\top \I A\} \vdash \top A^{n-1} \I B$.
  Since $\top A^{n-1}$ equals $A^{n-1}$, we may repeat the argument
  $n$-times to get $\Gamma \cup \{\top \I A\} \vdash \top \I B$.

  Conversely, let $\Gamma \cup \{\top \I A\} \vdash \top \I B$, i.e.,
  there is a proof $A_1 \I B_1,\ldots,A_m \I B_m$ of $\top \I B$ by
  $\Gamma \cup \{\top \I A\}$. By induction on the length of the proof,
  we show there is natural $n_i$ such that
  $\Gamma \vdash A^{n_i}A_i \I B_i$. Hence, (i) will result
  as a special case for $A_m \I B_m$ being $A \I B$.
  If $A_i \I B_i \in \Gamma$, then
  \vspace*{-.8\baselineskip}%
  \begin{align*}
    \CUT{A_i \I B_i}{\AX{B_iA \I B_i}}{AA_i \I B_i}
  \end{align*}
  proves that $\Gamma \vdash A^{n_i}A_i \I B_i$ for $n_i = 1$.
  If $A_i \I B_i$ is an instance of \Ax{} then so is $AA_i \I B_i$,
  i.e., $\Gamma \vdash A^{n_i}A_i \I B_i$ for $n_i = 1$.
  Finally, if $A_i \I B_i$ results
  from $A_j \I B_j$ and $A_k \I B_k$ ($j,k < i$) by \Cut, then using the
  induction hypothesis $\Gamma \vdash A^{n_j}A_j \I B_j$ and
  $\Gamma \vdash A^{n_k}A_k \I B_k$ for some natural $n_j$ and $n_k$.
  In addition to that, the fact that $A_i \I B_i$ results from $A_j \I B_j$
  and $A_k \I B_k$ by \Cut{} yields that $B_i = B_k$, $A_i = A_jC$ for some $C$,
  and $A_k = B_jC$. Then,
  \vspace*{-.8\baselineskip}
  \begin{align*}
    \CUT{\CUT{A^{n_j}A_j \I B_j}{\AX{B_jA^{n_k} \I B_jA^{n_k}}}
      {A^{n_j}A^{n_k}A_j \I A^{n_k}B_j}}
    {A^{n_k}B_jC \I B_i}
    {A^{n_j}A^{n_k}A_jC \I B_i}
  \end{align*}
  shows that $\Gamma \vdash A^{n_j}A^{n_k}A_jC \I B_i$, meaning that
  $\Gamma \vdash A^{n_i}A_i \I B_i$ for $n_i = n_j + n_k$.
\end{proof}

\begin{remark}
  Analogously as in the case of the rule of additivity, our logic does not
  admit a classic form of the deduction theorem. In other words,
  the exponent in Theorem~\ref{th:ded}\,(i) cannot be omitted.
\end{remark}

The semantic entailment can be formulated in terms of classes of
algebras other than integral commutative pomonoids.
For instance, we may define the notion of a model based on complete
residuated lattices and, as a consequence, obtain the notion of a semantic
entailment based on complete residuated lattices
and still be able to establish the completeness using the same
axiomatization. The completeness over complete residuated lattices shown
in the following assertion is an important observation because most of the
modern fuzzy logics use residuated lattices as structures
of degrees~\cite{CiHaNo1}.

\begin{theorem}[completeness over complete residuated lattices]\label{th:resl}%
  $\Gamma \vdash A \I B$ if{}f $A \I B$ is satisfied by each
  $\mathbf{L}$-model of $\Gamma$, where $\mathbf{L}$ is an arbitrary
  complete residuated lattice.
\end{theorem}
\begin{proof}
  The only-if part follows directly from the fact that from each
  complete residuated lattice
  $\mathbf{L} = \langle L,\wedge,\vee,\otimes,\rightarrow,0,1\rangle$ we
  can take its reduct $\langle L,\otimes,1\rangle$ and equip it with $\leq$
  defined by $a \leq b$ if{}f $a \rightarrow b = 1$. Clearly, 
  $\langle L,\leq,\otimes,1\rangle$ is a commutative integral pomonoid.
  Now, apply Theorem~\ref{th:snd}.

  In order to prove the if-part, it suffices to show that each commutative
  integral pomonoid can be embedded into a complete residuated lattice.
  The rest then follows by using Theorem~\ref{th:cmp}.
  Take any commutative integral pomonoid $\langle L,\leq,\otimes,1\rangle$.
  Consider the system $\mathcal{L}$ of all downward closed subsets of $L$
  with respect to $\subseteq$.
  It is well known that $\mathcal{L}$ with $\subseteq$ is
  a complete lattice. Put
  \begin{align*}
    X * Y &= \{z \in L;\,
    z \leq x \otimes y \text{ for some } x \in X \text{ and } y \in Y\}, \\
    X \rightarrow Y &= \{z \in L;\, X * \{z\} \subseteq Y\}.
  \end{align*}
  for any $X,Y \in \mathcal{L}$.
  Using the result of Galatos~\cite[Lemma~3.39]{GaJiKoOn:RL},
  $\mathbf{L} = \langle \mathcal{L},\cap,\cup,*,\rightarrow,\emptyset,L\rangle$
  is a complete residuated lattice and $h\!: L \to \mathcal{L}$ defined
  by $h(y) = \{x \in L;\, x \leq y\}$ is an embedding.
\end{proof}

We now turn our attention to the relationship of our rules and the classic FDs.
From the syntactic point of view, the classic FDs can be seen as MFDs in which
we allow to arbitrarily duplicate all occurrences of propositional variables.
From the semantic point of view, it turns out that FDs are just MFDs with
the semantics defined over the class of Boolean algebras. We show details
in the next theorem, where we use the following notation. For any $\Gamma$,
put
\begin{align}
  \GBool
  &= \Gamma \cup \{\alpha_p \I \alpha_p\alpha_p;\, p \in \mathrm{Var}\},
\end{align}
where $\alpha_p$ is defined as in~\eqref{eqn:Ap}. Now, we have:

\begin{theorem}[Boolean case extension]
  $\GBool \vdash A \I B$ if{}f\/
  $A \I B$ is satisfied by each $\mathbf{L}$-model of $\Gamma$,
  where $\mathbf{L}$ is the two-element Boolean algebra.
\end{theorem}
\begin{proof}
  The only-if part is easy to see since $\otimes$ in the two-element Boolean
  algebra is the truth function of the classic conjunction which is idempotent.
  In order to see the if-part, inspect the proof of Theorem~\ref{th:cmp} and
  observe that $E \equiv_{\GBool} E^n$ for any $E$
  and any natural $n$.
  Indeed, $\GBool \vdash E^n \I E$ follows from \Ax{}
  while $\GBool \vdash E \I E^n$ results by
  a repeated application of
  \vspace*{-.8\baselineskip}%
  \begin{align*}
    \CUT{E \I EE}{\CUT{E \I EE}{\AX{EEE \I EEE}}{EE \I EEE}}{E \I EEE}.
  \end{align*}
  Therefore, the operation $\circ$ in $\mathbf{L}/\GBool$ is idempotent
  and thus $\langle L/\GBool,\circ,[\top]_{\GBool}\rangle$ is a semilattice.
  In addition, we can show that $\leqslant_{\GBool}$ coincides with the
  meet-semilattice order induced by $\circ$. In order to see that, it suffices
  to show $[E]_{\GBool} \leqslant_{\GBool} [F]_{\GBool}$ if{}f
  $[E]_{\GBool} \circ [F]_{\GBool} = [E]_{\GBool}$. The latter condition
  can be rewritten as $[EF]_{\GBool} = [E]_{\GBool}$ which is true if{}f
  $EF \equiv_{\GBool} E$, i.e., if $\GBool \vdash EF \I E$
  and $\GBool \vdash E \I EF$. Since $EF \I E$ is an instance of \Ax,
  it suffices to check that $\GBool \vdash E \I F$ if{}f
  $\GBool \vdash E \I EF$ which is indeed the case:
  The if-part follows by
  \vspace*{-.8\baselineskip}%
  \begin{align*}
    \CUT{E \I EF}{\AX{EF \I F}}{E \I F}
  \end{align*}
  and the only-if part follows by
  \vspace*{-.8\baselineskip}%
  \begin{align*}
    \CUT{E \I EE}{\CUT{E \I F}{\AX{FE \I EF}}{EE \I EF}}{E \I EF}.
  \end{align*}
  As a consequence,
  if $\GBool \nvdash A \I B$ then there is an $\mathbf{L}/\GBool$-model $e$ of $\GBool$
  such that $e(A) \circ e(B) \ne e(A)$,
  where $\mathbf{L}/\GBool = \langle L/\GBool,\circ,[\top]_{\GBool}\rangle$
  is a meet-semilattice. Using standard arguments, $\mathbf{L}/\GBool$ can be
  embedded into a (complete) Boolean algebra $\mathbf{L}'$ of sets which is
  a subdirect product of two-element Boolean algebras~\cite{Bir:LT}.
  Hence, for the two-element Boolean algebra $\mathbf{L}$ on $\{0,1\}$
  with $0 < 1$ there must be an $\mathbf{L}$-evaluation $e$ which is
  a model of $\GBool$, $e(A) = 1$, and $e(B) = 0$, proving the claim.
\end{proof}

\section{Computational Issues}\label{sec:comput}
In this section, we discuss computational issues of the logic of
monoidal functional dependencies. We start by observing that the logic
is decidable and show that the provability in our logic may be expressed
as reducibility in an abstract rewriting system~\cite{Wec:UAfCS}. Based on
that, we show that for theories consisting of formulas in a special form,
there is a polynomial closure-like algorithm for deciding whether $A \I B$
is provable by a finite $\Gamma$.

\begin{theorem}\label{thm:dec}
  If $\Gamma$ is finite, then
  $\Gamma^\vdash = \{A \I B;\, \Gamma \vdash A \I B\}$ is decidable.
\end{theorem}
\begin{proof}
  Given a finite $\Gamma$, its deductive closure
  $\Gamma^\vdash = \{A \I B;\, \Gamma \vdash A \I B\}$
  is obviously recursively enumerable. In addition, using Theorem~\ref{th:resl}
  and the fact that the variety of residuated lattices has the
  \emph{finite embeddability property}~\cite{BlAl:FEP}
  (every finite partial residuated sublattice
  can be embedded into a finite residuated lattice) and therefore the
  \emph{strong finite model property} (every quasi-identity that fails in
  a residuated lattice fails in some finite one),
  we conclude that
  \begin{align*}
    \{A \I B;\, \Gamma \nmodels A \I B\} =
    \{A \I B;\, A \I B \not\in \Gamma^\vdash\}
  \end{align*}
  is recursively enumerable.
  As a consequence, it is decidable whether $A \I B$ is
  provable by a finite $\Gamma$.
\end{proof}

As a consequence of Theorem~\ref{thm:dec}, we obtain a naive approach
to decide whether $A \I B$ is provable by a finite $\Gamma$ which consists
in enumerating all proofs by $\Gamma$ and, simultaneously, generating
finite residuated lattices to find counterexamples. The enumeration of
proofs can be simplified since finding a proof of $A \I B$ may be seen
as a process in which we sequentially reduce $A$ in finitely many steps
using formulas in $\Gamma$. In order to formalize the rewriting process,
to each $\Gamma$ we associate a rewriting system
$\langle \mathcal{A},\rewrites\rangle$
where $\mathcal{A}$ is the set of all maps of the form~\eqref{eqn:A} and
$\rewrites$ is a binary relation on $\mathcal{A}$ such that
\begin{align}
  A \rewrites B
  \label{eqn:red}
\end{align}
for $A,B \in \mathcal{A}$ whenever the following conditions are satisfied:
\begin{enumerate}\parskip=0pt
\item
  $A = EG$ for some $E,G \in \mathcal{A}$,
\item
  $E \I F \in \Gamma$, and
\item
  $B = FG$.
\end{enumerate}
The transitive and reflexive closure $\reduces$ of
$\rewrites$ is called the \emph{reducibility} by $\Gamma$.
The basic relationship between the provability by $\Gamma$ and $\reduces$ is
described by the following assertion.

\begin{theorem}\label{th:rewr}
  $\Gamma \vdash A \I B$ if{}f
  there is $C \in \mathcal{A}$ such that $A \reduces BC$.
\end{theorem}
\begin{proof}
  Assume that there is $C \in \mathcal{A}$ such that $A \reduces BC$.
  By definition of $\reduces$, there are $A = D_0,\ldots,D_k = BC$ such that
  $D_0 \rewrites D_1 \rewrites \cdots \rewrites D_k$. By induction, assume
  that $\Gamma \vdash A \I D_i$ and observe that $D_i \rewrites D_{i+1}$
  means that $D_i = EG_i$ and $D_{i+1} = FG_i$ for some $E \I F \in \Gamma$
  and $G_i \in \mathcal{A}$. Therefore, from $A \I D_i$ we can
  infer $A \I D_{i+1}$ by $\Rwt$ and so $\Gamma \vdash A \I D_{i+1}$ because
  $\Rwt$ is a derived inference rule, cf. Remark~\ref{rem:inf}\,(c).
  Therefore, $\Gamma \vdash A \I D_k$ means $\Gamma \vdash A \I BC$ and so
  $\Gamma \vdash A \I B$ by $\Pro$.

  Conversely, we first argue that if $\Gamma \vdash A \I B$ then there is
  a proof $\varphi_1,\ldots,\varphi_n$ of $A \I B$ by $\Gamma$ which uses
  only the inference rules $\Ref$, $\Rwt$, and $\Pro$.
  In addition, we claim that the proof can be found so that
  the following additional properties are all satisfied:
  \begin{enumerate}\parskip=0pt
  \item
    $\varphi_1$ is $A \I A$ and it is the only instance of $\Ref$ in the proof;
  \item
    each $\varphi_i$ such that $1 < i< n$ is a formula in one of the following
    forms:
    \begin{enumerate}\parskip=0pt
    \item
      $\varphi_i \in \Gamma$, or
    \item
      $\varphi_i$ results by $\Rwt$ applied to some $\varphi_j$ ($j < i$)
      of the form $A \I X$ for some $X \in \mathcal{A}$ and
      a formula in $\Gamma$;
    \end{enumerate}
  \item
    $\varphi_n$ results from $\varphi_{n-1}$ by $\Pro$ and it is the only
    application of $\Pro$ used in the proof.
  \end{enumerate}
  Using the arguments in Remark~\ref{rem:inf}\,(c), there indeed is 
  a proof of $A \I B$ by $\Gamma$ which uses only $\Ref$, $\Pro$, and $\Rwt$.
  It remains to show that the proof may
  be transformed into a proof satisfying 1.--3. This can be shown using
  analogous arguments as in~\cite[Theorem 4.2]{Mai:TRD} which shows this
  in the classic setting with the rule of accumulation instead of $\Rwt$
  and proves the existence of the so-called RAP-derivation sequences,
  cf. also~\cite{Ma:Mcrdm}. A moment's reflection shows that the procedure
  in the proof of~\cite[Theorem 4.2]{Mai:TRD} may be carried over with 
  the weaker rule $\Rwt$ by performing the following steps during which we
  \begin{itemize}\parskip=0pt
  \item
    add $A \I A$ at the beginning of the proof (if it is not there);
  \item
    add an application of $\Pro$ at the end of the proof (if it is not there);
  \item
    eliminate all applications of $\Pro$ except for the last one using
    the argument that $\Pro$ commutes with $\Rwt$ and therefore a formula derived
    by first using $\Pro$ and then using $\Rwt$ may be derived by first
    using $\Rwt$ and then using $\Pro$;
  \item
    eliminate applications of $\Rwt$ which do not conform to either of
    (a) and (b) specified above by substituting each such an application by
    a series of applications of $\Rwt$ which yield formulas with $A$ as
    the antecedent and use only formulas in $\Gamma$.
    This can be done by going backwards
    through the proof and using the observation that
    \begin{align*}
      &\RWT{A \I DE}{\xRWT{E \I FG}{G \I H}{E \I FH}}{A \I DFH},
    \end{align*}
    can equivalently be expressed as
    \begin{align*}
      &\RWT{\xRWT{A \I DE}{E \I FG}{A \I DFG}}{G \I H}{A \I DFH},
    \end{align*}
    cf.~\cite[Theorem 4.2]{Mai:TRD}.
  \end{itemize}
  At this point we have shown that if $\Gamma \vdash A \I B$ then there is
  a proof $\varphi_1,\ldots,\varphi_n$ of $A \I B$ by $\Gamma$
  satisfying 1.--3. Let $A \I X_1,\ldots,A \I X_k$ be the subsequence of
  $\varphi_1,\ldots,\varphi_n$ which consists of all formulas
  with the antecedent $A$. By induction, we prove that
  $A \reduces X_i$ for all $i=1,\ldots,k$.
  We distinguish three cases.
  First, if $X_i = A$, then trivially $A \reduces X_i$.
  Second, if $A \I X_i \in \Gamma$,
  then directly by the definition of $\rewrites$, we get $A \rewrites X_i$
  and so $A \reduces X_i$. Third, if $A \I X_i$ results from $A \I X_j$ 
  (for some $j < i$) and some $E \I F \in \Gamma$ by $\Rwt$, then 
  $X_j = EG$ and $X_i = FG$ for some $G \in \mathcal{A}$ and so
  $X_j \rewrites X_i$, meaning $A \reduces X_j \rewrites X_i$, i.e.,
  $A \reduces X_i$. Altogether, $A \reduces X_i$ for all $i=1,\ldots,k$
  and as a special case for $i = k$, we get $A \reduces X_k = BC$
  for some $C \in \mathcal{A}$ because $A \I B$, being the last formula in
  $\varphi_1,\ldots,\varphi_n$, results from $A \I X_k$ by $\Pro$.
\end{proof}

Theorem~\ref{th:rewr} may be used to find proofs of $A \I B$ by
a finite $\Gamma$ in a more convenient way than the naive approach because
instead of storing proofs, one can just store representations of maps of
the form~\eqref{eqn:A} and in order to find a proof one may perform
a breadth-first search through a (possibly infinite) tree of derivations
starting with $A$. Needless to say, the procedure is still very expensive
because the memory consumed by the process can grow exponentially.
More importantly, in general it is still necessary to simultaneously
generate counterexamples in order to decide whether $A \I B$ follows
by $\Gamma$ because the search space is infinite.

In the rest of this section, we show that considerably more efficient
decision procedures may be found in case of theories consisting only
of particular formulas. We describe a procedure which exploits the
rewriting process and the result of Theorem~\ref{th:rewr} and which
resembles the well-known \textsc{Closure}
algorithm~\cite[Algorithm 4.2]{Mai:TRD}. We confine ourselves only
to so-called non-contracting theories.

A formula $A \Rightarrow B$ is called \emph{non-contracting} whenever $B$
can be written as $AC$ for some $C$. A theory $\Gamma$ is non-contracting
whenever all its formulas are non-contracting.

Clearly, if $\Gamma$ is non-contracting and $A \reduces B$,
then $A(y) \leq B(y)$ for all $y \in \Var$. From the point of view of
the inference rules, $\Rwt$ applied to non-contracting formulas acts
like the classic accumulation rule. In contrast to the classic properties
of closures of sets of attributes, there still is no guarantee that
for $A$ there is a greatest $B$ such that $A \reduces B$. Nevertheless,
for non-contracting theories, we may propose an algorithm as
in Figure~\ref{fig:alg} which generalize the well-known
algorithm \textsc{Member}~\cite[Algorithm 4.3]{Mai:TRD}.

\begin{figure}[p]
  \IncMargin{1em}%
  \SetKwInput{Input}{Input}%
  \SetKwInput{Output}{Output}%
  \begin{algorithm}[H]
    \Input{a finite non-contracting theory $\Gamma$ and
      a formula $A \I B$}
    \Output{boolean value}
    \BlankLine
    $\Delta \assign \Gamma \cup \{B \I By\}$\tcc*[r]{$y \in \Var$ is unused in $A,B,\Gamma$}
    $W \assign A$\tcc*[r]{$W$ is auxiliary map~\eqref{eqn:A}}
    $N \assign \textstyle\sum_{E \I F \in \Delta}\sum_{p \in \Var}E(p)$\!\!
    \tcc*[r]{counter}
    \BlankLine
    \Repeat{$L = W$ or $N \leq 0$ or $W(y) > 0$}{%
      $L \assign W$\tcc*[r]{$L$ is the last value of $W$}
      \ForEach{$E \I F \in \Delta$}{%
        \If{$W = EX$ for some $X$ of the form \eqref{eqn:A}}{%
          $W \assign FX$\tcc*[r]{update of $W$}
        }
      }
      $N \assign N - 1$\tcc*[r]{decrease the counter}
    }
    \BlankLine
    \eIf{$W(y) > 0$}{
      \Return{true};
    }{
      \Return{false};
    }
  \end{algorithm}
  \caption{Algorithm for deciding $\Gamma \vdash A \I B$
    for non-contracting $\Gamma$.}
  \label{fig:alg}
\end{figure}

The algorithm in Figure~\ref{fig:alg} accepts a finite non-contracting
theory $\Gamma$ and arbitrary formula $A \I B$ as its input. It is obvious
that the algorithm terminates after finitely many steps
(check the condition at line 12) and returns
a value \emph{true} or \emph{false}. The following assertion shows
that the algorithm decides $\vdash$.

\begin{theorem}
  The algorithm in Figure~\ref{fig:alg} is correct:
  For a non-contracting finite~$\Gamma$, the algorithm terminates after finitely
  many steps and returns ``true'' if{}f $\Gamma \vdash A \I B$.
\end{theorem}
\begin{proof}
  The algorithm uses $W$ as an auxiliary variable which represents
  a working multi-set in $\Var$ whose initial value is $A$ (see line~2).
  In addition, $\Delta$ is set to $\Gamma$ which is extended by
  a formula $B \I By$, see line 1, where $y$ is a fresh new propositional variable
  which does not appear in either formula in $\Gamma$ or in $A \I B$.
  Recall that using the abbreviated notation for~\eqref{eqn:Ap},
  for the consequent $By$ of $B \I By$ we have $By(y) = 1$ and $By(z) = B(z)$
  for all $z \ne y$.
  The algorithm utilizes an additional counter $N$ which is initially
  set to the total number of occurrences of propositional variables in all
  antecedents in $\Delta$, see line 3.

  The repeat-unit loop updates $W$ as long as it can be updated
  (the auxiliary variable $L$ is used to detect no update) based on the formulas
  in $\Delta$ and the property which is maintained after each update
  is that $A \reducesd W$. This is the same as in the ordinary \textsc{Closure}.
  Whenever an antecedent of a formula in $\Delta$ is contained in $W$,
  its consequent is added to $W$, see line 8.

  We now inspect the halting condition of the repeat-until loop.
  If $W(y) > 0$, it means that $B \I By$ has been used in line 8.
  Therefore, $A \reduces W$ such that $W = BX$ for some $X$
  and thus $\Gamma \vdash A \I B$ in which case
  the algorithm returns \emph{true}. If the repeat-until loop terminates
  and we have $W(y) = 0$, \emph{false} is returned. It suffices to show
  that in this case $\Gamma \nvdash A \I B$. To see this, observe that
  if $E \I F \in \Delta$ passes the condition in line 7, then it passes
  the condition in all consecutive iterations of the loop and $W$ is
  repeatedly updated by this formula (this is because all formulas
  in $\Delta$ are non-contracting, so the antecedent of $E \I F$
  cannot ``vanish'' from $W$). As a consequence, if $L \ne W$ holds
  when the algorithm reaches line 12 for the first time, then $L \ne W$
  for all consecutive iterations. Therefore, the repeat-until loop can be
  terminated because of $L = W$ only at the end of the first iteration
  in which case there is no formula in $\Delta$
  which may update the value of $W$ and so $\Gamma \nvdash A \I B$.

  Let us assume that $L \ne W$, $W(y) = 0$, and $N = 0$.
  We use the argument that $B \I By$ is either used to update $W$ (line 8)
  in $N$ steps with the initial value of $N$ given as in line 3,
  or it cannot be used to update $W$ at all. To see that, assume the
  worst case in which for $\Delta = \{E_1 \I F_1,\ldots,E_n \I F_n,B \I By\}$,
  only $E_1 \I F_1$ is used to update $W$ during the first $m_1$ iterations,
  then $E_1 \I F_1$ and $E_2 \I F_2$ are used simultaneously to update $W$
  during the next $m_2$ iterations, etc.,
  so that finally $B \I By$ is used to update~$W$. The key observation
  here is that $m_1$ cannot be strictly greater than the number of attributes in
  the antecedent of $E_2$ because in the worst case, the attributes (including
  their multiple occurrences) are added to $W$ one by one.
  That is, $m_1 \leq \sum_{p \in \Var}E_2(p)$ and analogously,
  $m_2 \leq \sum_{p \in \Var}E_3(p)$, etc.
  So, in the worst case, the use of $B \I By$ to update $W$ is bounded from
  above by
  \begin{align*}
    \textstyle
    \sum_{E \I F \in \Delta}\sum_{p \in \Var}E(p)
  \end{align*}
  iterations.
  As a conclusion, if $N$ initially set to the value in line 3 reaches $0$
  and $W(y) = 0$, there is no $X$ such that $A \reduces BX$,
  i.e., $\Gamma \nvdash A \I B$.
\end{proof}

\begin{remark}
  It is clear that the algorithm in Figure~\ref{fig:alg} is polynomial
  since it only represents an extension of \textsc{Closure} and \textsc{Member}
  which results in more iterations of the main loop than in
  the case of \textsc{Closure} but the number of iterations is bounded by
  the size of the input. In fact, our algorithm has quadratic worst-case
  time complexity, the same as \textsc{Closure}.
\end{remark}

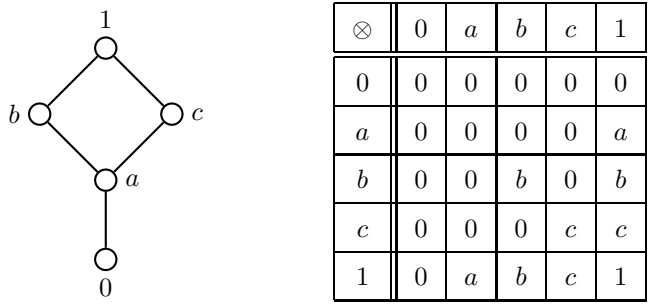
\begin{figure}
  \centering
  \begin{tabular}{@{}c@{\qquad\qquad}c@{}}
    \begin{tikzpicture}[inner sep=1mm, node distance=2.5em]
      \node [vertex] (jedna) [label=above:\ensuremath{1}] {};
      \node (aux) [below of=jedna] {};
      \node [vertex] (a) [below of=aux, label=right:\ensuremath{a}] {};
      \node [vertex] (b) [left of=aux, label=left:\ensuremath{b}] {}
      edge [edge, -] (jedna)
      edge [edge, -] (a);
      \node [vertex] (c) [right of=aux, label=right:\ensuremath{c}] {}
      edge [edge, -] (jedna)
      edge [edge, -] (a);
      \node [vertex] (nula) [below of=a, yshift=-.5em, label=below:\ensuremath{0}] {}
      edge [edge, -] (a);
    \end{tikzpicture}
    &
    \lower-5.5em\hbox{
      \bgroup
      \setlength{\tabcolsep}{1.6ex}
      \begin{tabular}{|c||c|c|c|c|c|}
        \hline
        $\otimes$ & $0$ & $a$ & $b$ & $c$ & $1$ \\
        \hline
        \hline
        $0$ & $0$ & $0$ & $0$ & $0$ & $0$ \\
        \hline
        $a$ & $0$ & $0$ & $0$ & $0$ & $a$ \\
        \hline
        $b$ & $0$ & $0$ & $b$ & $0$ & $b$ \\
        \hline
        $c$ & $0$ & $0$ & $0$ & $c$ & $c$ \\
        \hline
        $1$ & $0$ & $a$ & $b$ & $c$ & $1$ \\
        \hline
      \end{tabular}
      \egroup}
  \end{tabular}
  \caption{A non-linear partially ordered monoid.}
  \label{fig:nonlin}
\end{figure}

We conclude this section by a remark showing that if $\Gamma \nvdash A \I B$,
then it may not be possible to find a \emph{linear} $\mathbf{L}$-model of
$\Gamma$ which serves as a counterexample. A model is linear if the order
in $\mathbf{L}$ is total, i.e., for any $a,b \in L$, we have $a \leq b$
or $b \leq a$.

\begin{remark}\label{rem:linear}
  Take $\Gamma = \{p \I ux,\, p \I vy,\, uy \I q,\, vx \I q\}$. It can be
  easily seen that $\Gamma \vdash pp \I qq$ because
  $pp \rewrites uxp \rewrites uxvy \rewrites uqy \rewrites qq$.
  On the other hand, we have $\Gamma \nvdash p \I q$. Indeed, we can consider
  $\mathbf{L} = \langle L,\leq,\otimes,1\rangle$ with $\langle L,\leq\rangle$
  given by the Hasse diagram in Figure~\ref{fig:nonlin}\,(left) and
  with $\otimes$ given by the table in Figure~\ref{fig:nonlin}\,(right).
  For $e\!: \Var \to L$ such that 
  $e(p) = a$,
  $e(q) = 0$,
  $e(u) = b$,
  $e(v) = c$,
  $e(x) = b$,
  $e(y) = c$,
  we have
  \begin{align*}
    e(p) = a &\leq b = b \otimes b = e(ux), \\
    e(p) = a &\leq c = c \otimes c = e(vy), \\
    e(uy) = b \otimes c = 0 &\leq 0 = e(q), \\
    e(vx) = c \otimes b = 0 &\leq 0 = e(q),
  \end{align*}
  i.e., $e$ is an $\mathbf{L}$-model of $\Gamma$.
  In addition, $e(p) = a \nleq 0 = e(q)$, showing $\Gamma \nvdash p \I q$.
  We claim there is no linear $\mathbf{L}$-model of $\Gamma$ which refutes
  $p \I q$. Indeed, suppose that $e$ is a linear $\mathbf{L}$-model of $\Gamma$.
  Since $\mathbf{L}$ is linear, we have $e(x) \leq e(y)$ or $e(y) \leq e(x)$.
  In the first case, the monotony of $\otimes$ gives $e(ux) \leq e(uy)$ and
  so $e(p) \leq e(ux) \leq e(uy) \leq e(q)$, meaning $e \models p \I q$.
  In the second case, $e(p) \leq e(vy) \leq e(vx) \leq e(q)$,
  meaning $e \models p \I q$ again. Therefore, in the search
  for a counterexample, we cannot restrict ourselves
  to linear $\mathbf{L}$-models, only. It also means that our logic does not
  admit \emph{linear completions} of theories in the following sense:
  Given $\Gamma$ and $A \I B$ such that $\Gamma \nvdash A \I B$, 
  in general there is no
  $\Gamma' \supseteq \Gamma$ such that $\Gamma' \nvdash A \I B$
  and $\Gamma' \vdash E \I F$ or $\Gamma' \vdash F \I E$ for all $E$ and $F$
  of the form~\eqref{eqn:A}. As a further consequence, our logic does not
  admit the principle of ``proofs by cases'': In general the facts that
  $\Gamma \cup \{E \I F\} \vdash A \I B$ and
  $\Gamma \cup \{F \I E\} \vdash A \I B$ do not yield 
  $\Gamma \vdash A \I B$. This also explains our choice of the name for the
  logic. Namely, our choice of the word ``monoidal'' over the word ``fuzzy''
  because in the modern understanding of (formal) fuzzy logics, properties
  like the presence of the principle of proofs by cases are considered
  essential, see \cite{CiHaNo1} for details.
\end{remark}

\section{Propositional vs. Relational Semantics}\label{sec:propdb}
So far, we have used a \emph{propositional semantics} of the formulas.
That means, MFDs have been interpreted given evaluations of propositional
variables. In order to establish the desired connection to relational
databases, we show that MFDs have an equivalent semantics based on
evaluating MFDs in relations on relation schemes. Since relations in
databases are considered on finite relation schemes, we consider here
only entailment from finite theories.

Let $\mathbf{L} = \langle L,\leq,\otimes,1\rangle$ be an integral
commutative pomonoid.
Let $R$ be a relation scheme (a finite set of attributes);
$\mathbf{r}$ be a relation on $R$ in the usual sense;
$D^{\mathbf{r}}_p$ denote the domain of attribute $p$ in~$\mathbf{r}$
(we consider the notion of a domain as a synonym for the notion of a type,
see~\cite{DaDa:DTRM}).
Furthermore, consider for any $p \in R$
a map $\approx^{\mathbf{r}}_p: D^{\mathbf{r}}_p \times D^{\mathbf{r}}_p \to L$,
where $L$ is the support of $\mathbf{L}$.
Following the discussion in Section~\ref{sec:intro},
the result of $d_1 \approx^{\mathbf{r}}_p d_2$ can be seen as a degree
in $L$ which is an answer to
the atomic query: ``Is $d_1$ similar to $d_2$?''
We assume that $\approx^{\mathbf{r}}_p$ are supplied along with the data
and assume that $d \approx^{\mathbf{r}}_p d = 1$ for each
$d \in D^{\mathbf{r}}_p$ and $p \in \Var$ (i.e., each element is similar
to itself to degree $1$---the highest degree in $\mathbf{L}$).

For $\mathbf{r}$, $A$ of the form~\eqref{eqn:A},
and any tuples $r_1,r_2 \in \mathbf{r}$, we put
\begin{align}
  r_1 \approx^{\mathbf{r}}_A r_2 &=
  \bigl(r_1(p_1) \approx^{\mathbf{r}}_{p_1} r_2(p_1)\bigr)^{A(p_1)}
  \otimes \cdots \otimes
  \bigl(r_1(p_n) \approx^{\mathbf{r}}_{p_n} r_2(p_n)\bigr)^{A(p_n)}
  \label{eqn:sim}
\end{align}
for $R \subseteq \{p_1,\dots,p_n\}$. Since $\otimes$ serves as an
interpretation of a conjunction, \eqref{eqn:sim} can be seen as a degree
in $L$ which is a result of conjunctive query:
``Are $r_1(p_1)$ similar to $r_2(p_1)$ \emph{and} $\cdots$ \emph{and}
$r_1(p_n)$ similar to $r_2(p_n)$?'' Therefore, $r_1 \approx^{\mathbf{r}}_A r_2$
is the degree to which tuples $r_1$ and $r_2$ in $\mathbf{r}$ are similar on
all attributes in $A$.
For $\mathbf{r}$ and $A \I B$ we say that
\emph{$\mathbf{r}$ satisfies $A \I B$}, written $\mathbf{r} \models A \I B$,
if for any tuples $r_1,r_2 \in \mathbf{r}$, the following inequality holds:
\begin{align}
  r_1 \approx^{\mathbf{r}}_A r_2 &\leq r_1 \approx^{\mathbf{r}}_B r_2.
  \label{eqn:db_sem}
\end{align}
Using the notion of satisfaction of MFDs in relations, we introduce models
and semantic entailment as before. Namely, we put
\begin{align}
  \mathrm{Mod}(\Gamma) =
  \{\mathbf{r};\, \mathbf{r} \models E \I F \text{ for all } E \I F \in \Gamma\}
\end{align}
and call each $\mathbf{r} \in \mathrm{Mod}(\Gamma)$
a (\emph{relational}) \emph{model of $\Gamma$.}
An MFD $A \I B$ is \emph{semantically entailed by $\Gamma$}
(\emph{in the relational sense}) if 
$\mathrm{Mod}(\Gamma) \subseteq \mathrm{Mod}(\{A \I B\})$, i.e.,
if $A \I B$ is satisfied in every relational model of $\Gamma$.

\begin{theorem}\label{th:altsem}
  Let $\Gamma$ be finite. Then,
  $\Gamma \models A \I B$
  if{}f\/
  $A \I B$ is 
  semantically entailed by $\Gamma$
  in the relational sense.
\end{theorem}
\begin{proof}
  Let $R$ be a finite subset of $\Var$ which contains all propositional
  variables appearing in $A \I B$ and all formulas in $\Gamma$.
  The if-part follows by the fact that for each $\mathbf{L}$-model
  $e$ of $\Gamma$ there is $\mathbf{r} \in \mathrm{Mod}(\Gamma)$ such that
  $e \models E \I F$ if{}f $\mathbf{r} \models E \I F$ for any $E \I F$.
  Namely, we can consider $\mathbf{r} = \{r_1,r_2\}$ such that
  $r_1(p) = 1$ for any $p \in R$, $r_2(p) = e(p)$,
  and $1 \approx^{\mathbf{r}}_{p} e(p) =
  e(p) \approx^{\mathbf{r}}_{p} 1 =
  e(p)$ for any $p \in R$.
  Hence, the domains of attributes in $\mathbf{r}$ are considered as
  subsets of $L$.
 
  Conversely, for each $\mathbf{r} \in \mathrm{Mod}(\Gamma)$ with
  all $\approx^{\mathbf{r}}_p$ defined using $\mathbf{L}$, there
  is a finite set $\mathcal{S}$ of $\mathbf{L}$-models $e$ such that
  $\mathbf{r} \models E \I F$ if{}f $e \models E \I F$ for all
  $e \in \mathcal{S}$. In particular,
  we let $\mathcal{S} = \{e_{r_1,r_2};\, r_1,r_2 \in \mathbf{r}\}$,
  where $e_{r_1,r_2}(p) = r_1(p) \approx^{\mathbf{r}}_{p} r_2(p)$ for
  all $p \in R$. The rest is easy to check.
\end{proof}

As a result of Theorem~\ref{th:altsem}, the relational and propositional
semantics have the same notion of semantic entailment and thus all
observations on provability we have made in Section~\ref{sec:log},
Section~\ref{sec:compl}, and Section~\ref{sec:comput} apply
to both semantics.

We conclude the paper by an illustrative example in which we return to our
initial motivation presented in Section~\ref{sec:intro}. The example shows
a particular relation with similarities on domains and examples of constraints
formulated in terms of MFDs. In addition, we show how the inference system
can be used to reason about dependencies which hold in data.

\begin{figure}
  \centering
  \renewcommand{\arraystretch}{0.9}
  \begin{tabular}{|*{3}c|}
    \hline
    \atr{AREA} &
    \atr{LOCATION} &
    \atr{PRICE} \\
    \hline
    $2,510$ & $[12.2, 23.4]$ & $810,000$ \\
    $2,730$ & $[35.3, 40.0]$ & $650,000$ \\
    $2,850$ & $[95.8, 82.3]$ & $625,000$ \\
    $4,250$ & $[20.1, 45.7]$ & $925,000$ \\
    \hline
  \end{tabular}
  \caption{Illustrative relation $\mathbf{r}$ on relation scheme
    $\{\atr{AREA},\atr{LOCATION},\atr{PRICE}\}$.}
  \label{fig:r}
\end{figure}

\begin{example}
  Consider the relation $\mathbf{r}$ in Figure~\ref{fig:r}. The relation
  is defined on relation scheme consisting of attributes $\atr{AREA}$,
  $\atr{LOCATION}$, and $\atr{PRICE}$ representing property area measured
  in square feet, location represented by coordinates on a fictitious map,
  and price in USD. For illustrative purposes, we introduce similarities
  on domains of the attributes by transforming the Euclidian distance
  of domain values to the real unit interval using the exponential function.
  In particular, denoting the Euclidian distance of $a$ and $b$
  by $\mathrm{d}(a,b)$,
  for $y$ being $\atr{AREA}$, $\atr{LOCATION}$, or $\atr{PRICE}$,
  we put
  \begin{align*}
    a \approx^\mathbf{r}_y b &=
    \textstyle \exp\bigl(- 10^{-c_y} \cdot \mathrm{d}(a,b)\bigr),
  \end{align*}
  where $c_\atr{LOCATION} = 2$, $c_\atr{AREA} = 4$, and $c_\atr{PRICE} = 6$.
  Furthermore, we assume that $\mathbf{L}$ is the same as
  in Remark~\ref{rem:nsound}. In this setting, we have
  \begin{align*}
    \mathbf{r} \models (\atr{LOCATION} \logand \atr{AREA}) \I \atr{PRICE},
  \end{align*}
  Indeed, in the non-trivial cases and considering
  the symmetry of our similarity, we get that
  \begin{align*}
    [12.2, 23.4] \approx^\mathbf{r}_{\atr{L}} [35.3, 40.0] \otimes
    2510 \approx^\mathbf{r}_{\atr{A}} 2730 = 0.73
    &\leq 0.85 = 810000 \approx^\mathbf{r}_{\atr{P}} 650000, \\
    [12.2, 23.4] \approx^\mathbf{r}_{\atr{L}} [95.8, 82.3] \otimes
    2510 \approx^\mathbf{r}_{\atr{A}} 2850 = 0.34
    &\leq 0.83 = 810000 \approx^\mathbf{r}_{\atr{P}} 625000, \\
    [12.2, 23.4] \approx^\mathbf{r}_{\atr{L}} [20.1, 45.7] \otimes
    2510 \approx^\mathbf{r}_{\atr{A}} 4250 = 0.66
    &\leq 0.89 = 810000 \approx^\mathbf{r}_{\atr{P}} 925000, \\
    [35.3, 40.0] \approx^\mathbf{r}_{\atr{L}} [95.8, 82.3] \otimes
    2730 \approx^\mathbf{r}_{\atr{A}} 2850 = 0.47
    &\leq 0.97 = 650000 \approx^\mathbf{r}_{\atr{P}} 625000, \\
    [35.3, 40.0] \approx^\mathbf{r}_{\atr{L}} [20.1, 45.7] \otimes
    2730 \approx^\mathbf{r}_{\atr{A}} 4250 = 0.73
    &\leq 0.75 = 650000 \approx^\mathbf{r}_{\atr{P}} 925000, \\
    [95.8, 82.3] \approx^\mathbf{r}_{\atr{L}} [20.1, 45.7] \otimes
    2850 \approx^\mathbf{r}_{\atr{A}} 4250 = 0.37
    &\leq 0.74 = 625000 \approx^\mathbf{r}_{\atr{P}} 925000,
  \end{align*}
  where \atr{A}, \atr{L}, and \atr{P} are abbreviations for
  \atr{AREA}, \atr{LOCATION}, and \atr{PRICE}, respectively.
  Therefore, for this particular $\mathbf{r}$ and the choice of
  the similarities on domains and~$\otimes$, the dependency says that
  similar values of area and location imply similar prices. In contrast,
  \begin{align*}
    \mathbf{r} \nmodels \atr{PRICE} \I \atr{LOCATION}
  \end{align*}
  because we have, e.g.,
  \begin{align*}
    810000 \approx^\mathbf{r}_{\atr{P}} 625000 = 0.83
    &\nleq 0.35 = [12.2, 23.4] \approx^\mathbf{r}_{\atr{L}} [95.8, 82.3].
  \end{align*}
  In words, similar prices do not yield similar locations. Therefore,
  $\atr{PRICE} \I \atr{LOCATION}$ is an example of a dependency
  which is trivially satisfied in $\mathbf{r}$ as an ordinary FD but it
  is not satisfied in $\mathbf{r}$ as an MFD.
  Let us now assume a situation of violating
  $(\atr{LOCATION} \logand \atr{AREA}) \I \atr{PRICE}$ by an attempted
  insertion of a new tuple whose values of 
  $\atr{AREA}$, $\atr{LOCATION}$, and $\atr{PRICE}$ are
  $2,600$, $[50.0, 50.0]$, and $\$450,000$, respectively.
  For this tuple, we have
  \begin{align*}
    [50.0, 50.0] \approx^\mathbf{r}_{\atr{L}} [35.3, 40.0] \otimes
    2600 \approx^\mathbf{r}_{\atr{A}} 2730 &= 0.8263 \\
    &\nleq 0.8187 = 450000 \approx^\mathbf{r}_{\atr{P}} 650000, \\
    [50.0, 50.0] \approx^\mathbf{r}_{\atr{L}} [20.1, 45.7] \otimes
    2600 \approx^\mathbf{r}_{\atr{A}} 4250 &= 0.6268 \\
    &\nleq 0.6219 = 450000 \approx^\mathbf{r}_{\atr{P}} 925000,
  \end{align*}
  i.e., $(\atr{LOCATION} \logand \atr{AREA}) \I \atr{PRICE}$ would be
  violated. 
  Now, let us assume that despite the constraint violation, we would like
  to insert the tuple in $\mathbf{r}$ because the constraint given by 
  $(\atr{LOCATION} \logand \atr{AREA}) \I \atr{PRICE}$
  was recognized as too strict. In such a situation, we have the following
  options to deal with the validity of the constraint: We may
  \begin{enumerate}\parskip=0pt
  \item[(a)]
    replace $\otimes$ by a more suitable aggregation function,
  \item[(b)]
    redefine similarities on domains, or
  \item[(c)]
    replace the MFD by a weaker constraint.
  \end{enumerate}
  By applying (a) and (b), we may render
  $(\atr{LOCATION} \logand \atr{AREA}) \I \atr{PRICE}$ valid in the
  new relation, however, the change of $\otimes$ (i.e., the change of
  $\mathbf{L}$) or similarities on domains may not be desirable because
  there may be other dependencies where the present choice is adequate
  and works well. Following (c) means introducing a new MFD instead
  of $(\atr{LOCATION} \logand \atr{AREA}) \I \atr{PRICE}$ which 
  can be derived from $(\atr{LOCATION} \logand \atr{AREA}) \I \atr{PRICE}$
  using the inference system presented in Section~\ref{sec:log}.
  For instance, we may want to put less emphasis on
  the similarity of areas, i.e., we may replace the constraint by
  \begin{align*}
    (\atr{LOCATION} \logand \atr{AREA} \logand \atr{AREA}) \I \atr{PRICE}
  \end{align*}
  which is satisfied in the new relation because
  \begin{align*}
    [50.0, 50.0] \approx^\mathbf{r}_{\atr{L}} [35.3, 40.0] &\otimes
    2600 \approx^\mathbf{r}_{\atr{A}} 2730 \otimes
    2600 \approx^\mathbf{r}_{\atr{A}} 2730 = 0.8156 \\
    & \leq 0.8187 = 450000 \approx^\mathbf{r}_{\atr{P}} 650000, \\
    [50.0, 50.0] \approx^\mathbf{r}_{\atr{L}} [20.1, 45.7] &\otimes
    2600 \approx^\mathbf{r}_{\atr{A}} 4250 \otimes
    2600 \approx^\mathbf{r}_{\atr{A}} 4250 = 0.5315 \\
    & \leq 0.6219 = 450000 \approx^\mathbf{r}_{\atr{P}} 925000.
  \end{align*}
  Let us note that $(\atr{LOCATION} \logand \atr{AREA} \logand \atr{AREA}) \I
  \atr{PRICE}$ is indeed derivable from the original constraint:
  \vspace*{-.6\baselineskip}%
  \begin{align*}
    \CUT{(\atr{LOCATION} \logand \atr{AREA}) \I \atr{PRICE}}{
      \AX{(\atr{PRICE} \logand \atr{AREA}) \I \atr{PRICE}}}{
      (\atr{LOCATION} \logand \atr{AREA} \logand \atr{AREA}) \I \atr{PRICE}}.
  \end{align*}
  As a result,
  $(\atr{LOCATION} \logand \atr{AREA} \logand \atr{AREA}) \I \atr{PRICE}$
  may be seen as prescribing the same dependency as 
  $(\atr{LOCATION} \logand \atr{AREA}) \I \atr{PRICE}$
  (i.e., similar locations and areas imply similar prices) except that we
  put less emphasis on the similarity of areas. This example illustrates that
  MFDs with multiple occurrences of attributes in the antecedent (or consequent)
  are not only what we inevitably get when we shift from idempotent aggregation
  functions to general ones but can also be used to control the sensitivity
  of similarity-based constraints.
\end{example}

\section{Related Work}\label{sec:relwork}
In this section, we comment on the relationship of our logic to other approaches
which study formulas expressing if-then dependencies whose semantics involves
degrees coming from general structures of truth values.

First, let us note that there exists a vast amount of papers on
``fuzzy functional dependencies'', often with questionable technical quality,
which combine (in various ways) the concepts of fuzzy sets and functional
dependencies in order
to formalize vague dependencies between attributes. While this idea is tempting
and close to what we present here, our objection is that most of these papers
are purely definitional or just experimental and are not interested in the
underlying logic in the narrow sense of it (i.e., in logic as a study of consequence).
From one viewpoint this is not surprising since a number of papers in this
category predate the beginning of systematic formalization of various types
of fuzzy logics which appeared in the late 90's, see~\cite{Haj:MFL} as
a standard reference and a historical overview. One of the most influential
early approaches that enjoyed interest in the database community
is~\cite{RaMa:Ffdljdfrds},
further papers dealing with fuzzy functional dependencies and related phenomena
include \cite{BuPe:Frdrd,CuVi:Ndffdfrd,PrTe:Gdratiuivq}.
Since our paper is not a survey,
we do not write further details on such approaches and refer interested
readers to~\cite{BeVy:Codd} where they can find further comments.

Our approach is also related to approaches to graded if-then rules which
are motivated by formal concept analysis~\cite{GaWi:FCA} of data with
graded attributes. In~\cite{Po:FB}, Polland proposed graded if-then rules
with semantics defined using complete residuated lattices as structures of
degrees. The approach has been later extended and more developed
in~\cite{BeVy:ADfDwG} by considering formalizations of
linguistic hedges~\cite{EsGoNo:Hedges,Za:Afstilh} as additional parameters
of semantics of the if-then rules. Compared to the present paper, there are
significant technical and epistemic differences.
First, the approaches in~\cite{BeVy:ADfDwG,Po:FB} use arbitrary,
but fixed, structures of degrees. That is, instead of focusing on formulas
which may be true in $\mathbf{L}$-models where $\mathbf{L}$ ranges over
a class of structures of degrees (like the class of all integral
commutative pomonoids), the papers fix $\mathbf{L}$ and define semantics
with respect to the fixed $\mathbf{L}$. Second, the formulas
in~\cite{BeVy:ADfDwG,Po:FB} are syntactically different. Namely,
they involve idempotent conjunctions instead of general non-idempotent ones.
On the other hand, the formulas use degrees in $\mathbf{L}$ to express lower
bound of degrees to which attributes in antecedents and consequents of
formulas are present---this is possible because $\mathbf{L}$ is fixed.
As a consequence, the formulas in~\cite{BeVy:ADfDwG,Po:FB} allow to express
dependencies like ``if $x$ is true at least to degree $a$ and $y$ is true
at least to degree $b$, then $z$ is true at least to degree $c$'' with
$a,b,c$ being degrees in the fixed $\mathbf{L}$. Third, unlike our logic,
the logic for such rules is Pavelka-style
complete~\cite{Pav:Ofl1,Pav:Ofl2,Pav:Ofl3} which means that degrees
of semantic consequence agree with (suitably defined) degrees of provability.
In our case, Pavelka-style completeness cannot be considered because
$\mathbf{L}$ is not fixed. On the other hand, \cite{BeVy:ADfDwG} shows that
in order to obtain Pavelka-style completeness for a general (infinite)
$\mathbf{L}$, one has to resort to admitting infinitary inference rules
which is not our case.

There exist approaches to if-then dependencies in relational
databases which are based on the notion of a similarity considered as
a classic relation which is at least reflexive and symmetric.
The approaches are developed in the context of the classic Boolean logic.
Examples of such approaches include the matching dependencies
\cite{Fan:Dcrm,Fan:Rarmr} which formalize constraints for matching records
from unreliable data sources, cf. also~\cite{FaGe:FDQM}.

Note that recently, probabilistic databases~\cite{DaReSu:Pdditd} aiming
at representation and querying of uncertain data are gaining popularity.
Our approach is not directly related because it
does not involve uncertainty in the probabilistic sense---like in
the classic relational model, our data is certain. Also note that the degrees
(the elements of integral commutative pomonoids) we use are not and shall not be
interpreted as degrees of belief or evidence (even if $L = [0,1])$,
cf. ``the frequentist's temptation'' in \cite{Haj:MFL} and also \cite{HaPa:Adfl}.

\section{Conclusion and Open Problems}\label{sec:concl}
We have introduced a logic for monoidal functional dependencies (MFDs) and we
proved the logic is complete with respect to the class of all integral
commutative partially ordered monoids. In addition, we have shown completeness
with respect to all complete residuated lattices. The logic of the classic FDs
may be seen as an extension of the logic of MFDs which consists of adding formulas
expressing the idempotency of conjunction. It has two natural
semantics---propositional one and relational one. We have shown the logic is
decidable and in case of non-contracting theories there is a polynomial
algorithm for deciding whether a formula follows by a finite set of other
formulas.

Further issues we consider worth studying include:
\begin{itemize}
\item
  methods for extracting non-redundant bases consisting of formulas which entail
  all formulas true in given data as in~\cite{GuDu};
\item
  approaches to use MFDs as association rules,
  possible descriptions of non-redundant rules and related algorithms,
  cf.~\cite{Zak:Mnrar};
\item
  algorithms for deciding entailment of formulas which are
  not limited to non-contracting theories;
\item
  detailed analysis of computational complexity of algorithms related to MFDs
  such as the algorithm for the entailment problem, establishing lower and upper
  complexity bounds for the entailment problem;
\item
  further logical and model-theoretical properties, e.g.,
  characterization of model classes by closure properties,
  extensions of the logic including completeness over classes
  of (linear) algebras which appear in fuzzy logics~\cite{CiHaNo1,EsGo:MTL,Haj:MFL};
\item
  possible generalizations which take into account more general structures
  than the integral commutative partially ordered monoids (e.g., structures with
  non-commutative or non-associative aggregation operations).
\end{itemize}

\subsubsection*{Acknowledgment}
Supported by grant no. \verb|P202/14-11585S| of the Czech Science Foundation.


\footnotesize\openup-1pt
\bibliographystyle{amsplain}
\bibliography{monfds}

\providecommand{\bysame}{\leavevmode\hbox to3em{\hrulefill}\thinspace}
\providecommand{\MR}{\relax\ifhmode\unskip\space\fi MR }
\providecommand{\MRhref}[2]{%
  \href{http://www.ams.org/mathscinet-getitem?mr=#1}{#2}
}
\providecommand{\href}[2]{#2}
\begin{thebibliography}{10}

\bibitem{Arm:Dsdbr}
William~Ward Armstrong, \emph{Dependency structures of data base
  relationships}, Information Processing 74: Proceedings of IFIP Congress
  (Amsterdam) (J.~L. Rosenfeld and H.~Freeman, eds.), North Holland, 1974,
  pp.~580--583.

\bibitem{BeBe:Cprttdonfrs}
Catriel Beeri and Philip~A. Bernstein, \emph{Computational problems related to
  the design of normal form relational schemas}, ACM Trans. Database Syst.
  \textbf{4} (1979), 30--59.

\bibitem{Bel:FRS}
Radim Belohlavek, \emph{Fuzzy {R}elational {S}ystems: {F}oundations and
  {P}rinciples}, Kluwer Academic Publishers, Norwell, MA, USA, 2002.

\bibitem{BeVy:Codd}
Radim Belohlavek and Vilem Vychodil, \emph{Codd's relational model from the
  point of view of fuzzy logic}, J. Log. Comput. \textbf{21} (2011), no.~5,
  851--862.

\bibitem{BeVy:ADfDwG}
\bysame, \emph{Attribute dependencies for data with grades}, CoRR
  \textbf{\href{http://arxiv.org/abs/1402.2071}{abs/1402.2071}} (2014).

\bibitem{Bir:LT}
Garrett Birkhoff, \emph{Lattice theory}, 1st ed., American Mathematical
  Society, Providence, 1940.

\bibitem{BlAl:FEP}
Willem~J. Blok and Clint~J. Van~Alten, \emph{The finite embeddability property
  for residuated lattices, pocrims and {BCK}-algebras}, Algebra Universalis
  \textbf{48} (2002), no.~3, 253--271.

\bibitem{BuPe:Frdrd}
Billy~P. Buckles and Frederick~E. Petry, \emph{A fuzzy representation of data
  for relational databases}, Fuzzy Sets and Systems \textbf{7} (1982), no.~3,
  213--226.

\bibitem{CiHaNo1}
Petr Cintula, Petr H\'ajek, and Carles Noguera (eds.), \emph{Handbook of
  {M}athematical {F}uzzy {L}ogic, {V}olume 1}, Studies in Logic, Mathematical
  Logic and Foundations, vol.~37, College Publications, 2011.

\bibitem{CuVi:Ndffdfrd}
Juan~C. Cubero and Mar\'{\i}a~Amparo Vila, \emph{A new definition of fuzzy
  functional dependency in fuzzy relational databases}, International Journal
  of Intelligent Systems \textbf{9} (1994), 441--448.

\bibitem{DaReSu:Pdditd}
Nilesh Dalvi, Christopher R\'{e}, and Dan Suciu, \emph{Probabilistic databases:
  diamonds in the dirt}, Commun. ACM \textbf{52} (2009), 86--94.

\bibitem{DaDa:DTRM}
Christopher~J. Date and Hugh Darwen, \emph{Databases, types, and the relational
  model: The third manifesto}, 3rd ed., Addison-Wesley, 2006.

\bibitem{DeCa}
Claude Delobel and Richard~G. Casey, \emph{Decomposition of a data base and the
  theory of boolean switching functions}, IBM Journal of Research and
  Development \textbf{17} (1973), no.~5, 374--386.

\bibitem{Di38}
Robert~P. Dilworth, \emph{Abstract residuation over lattices}, Bull. Amer.
  Math. Soc. \textbf{44} (1938), 262--268.

\bibitem{EsGo:MTL}
Francesc Esteva and Llu{\'\i}s Godo, \emph{Monoidal t-norm based logic:
  {T}owards a logic for left-continuous t-norms}, Fuzzy Sets and Systems
  \textbf{124} (2001), no.~3, 271--288.

\bibitem{EsGoNo:Hedges}
Francesc Esteva, Llu\'\i s~Godo, and Carles Noguera, \emph{A logical approach
  to fuzzy truth hedges}, Information Sciences \textbf{232} (2013), 366--385.

\bibitem{Fagin}
Ronald Fagin, \emph{Functional dependencies in a relational database and
  propositional logic}, IBM Journal of Research and Development \textbf{21}
  (1977), no.~6, 534--544.

\bibitem{Fa98:CFIfMS}
\bysame, \emph{Combining fuzzy information from multiple systems}, J. Comput.
  Syst. Sci. \textbf{58} (1999), no.~1, 83--99.

\bibitem{Fagin:Oaafm}
Ronald Fagin, Amnon Lotem, and Moni Naor, \emph{Optimal aggregation algorithms
  for middleware}, J. Comput. Syst. Sci. \textbf{66} (2003), no.~4, 614--656.

\bibitem{Fan:Dcrm}
Wenfei Fan, Hong Gao, Xibei Jia, Jianzhong Li, and Shuai Ma, \emph{Dynamic
  constraints for record matching}, The VLDB Journal \textbf{20} (2011), no.~4,
  495--520.

\bibitem{FaGe:FDQM}
Wenfei Fan and Floris Geerts, \emph{Foundations of {D}ata {Q}uality
  {M}anagement}, Synthesis Lectures on Data Management, vol.~4, Morgan \&
  Claypool Publishers, 2012.

\bibitem{Fan:Rarmr}
Wenfei Fan, Xibei Jia, Jianzhong Li, and Shuai Ma, \emph{Reasoning about record
  matching rules}, Proc. VLDB Endow. \textbf{2} (2009), no.~1, 407--418.

\bibitem{GaJiKoOn:RL}
Nikolaos Galatos, Peter Jipsen, Tomacz Kowalski, and Hiroakira Ono,
  \emph{{R}esiduated {L}attices: {A}n {A}lgebraic {G}limpse at {S}ubstructural
  {L}ogics, {V}olume 151}, 1st ed., Elsevier Science, San Diego, USA, 2007.

\bibitem{GaWi:FCA}
Bernhard Ganter and Rudolf Wille, \emph{Formal concept analysis: Mathematical
  foundations}, 1st ed., Springer-Verlag New York, Inc., Secaucus, NJ, USA,
  1997.

\bibitem{Gog:Lic}
Joseph~A. Goguen, \emph{The logic of inexact concepts}, Synthese \textbf{19}
  (1979), 325--373.

\bibitem{GuDu}
Jean-Louis Guigues and Vincent Duquenne, \emph{Familles minimales
  d'im\-pli\-ca\-tions informatives resultant d'un tableau de donn\'ees
  binaires}, Math. Sci. Humaines \textbf{95} (1986), 5--18.

\bibitem{Haj:MFL}
Petr H\'ajek, \emph{Metamathematics of {F}uzzy {L}ogic}, Kluwer Academic
  Publishers, Dordrecht, The Netherlands, 1998.

\bibitem{HaPa:Adfl}
Petr H\'ajek and Jeff Paris, \emph{A dialogue on fuzzy logic}, Soft Computing
  \textbf{1} (1997), no.~1, 3--5.

\bibitem{Ho:ML}
Ulrich H\"ohle, \emph{Monoidal logic}, Fuzzy-Systems in Computer Science
  (R.~Kruse, J.~Gebhardt, and R.~Palm, eds.), Artificial Intelligence /
  K\"unstliche Intelligenz, Vieweg+Teubner Verlag, 1994, pp.~233--243.

\bibitem{Il:ASOTQPTiRDS}
Ihab~F. Ilyas, George Beskales, and Mohamed~A. Soliman, \emph{A survey of top-k
  query processing techniques in relational database systems}, ACM Comp. Surv.
  \textbf{40} (2008), no.~4, 11:1--11:58.

\bibitem{KMP:TN}
Erich~Peter Klement, Radko Mesiar, and Endre Pap, \emph{Triangular {N}orms}, 1
  ed., Springer, 2000.

\bibitem{Ma:Mcrdm}
David Maier, \emph{Minimum covers in relational database model}, J. ACM
  \textbf{27} (1980), no.~4, 664--674.

\bibitem{Mai:TRD}
\bysame, \emph{Theory of {R}elational {D}atabases}, Computer Science Pr,
  Rockville, MD, USA, 1983.

\bibitem{Me87}
Elliott Mendelson, \emph{Introduction to {M}athematical {L}ogic}, Chapman and
  Hall, 1987.

\bibitem{Pav:Ofl1}
Jan Pavelka, \emph{On fuzzy logic {I}: {M}any-valued rules of inference},
  Mathematical Logic Quarterly \textbf{25} (1979), no.~3--6, 45--52.

\bibitem{Pav:Ofl2}
\bysame, \emph{On fuzzy logic {II}: {E}nriched residuated lattices and
  semantics of propositional calculi}, Mathematical Logic Quarterly \textbf{25}
  (1979), no.~7--12, 119--134.

\bibitem{Pav:Ofl3}
\bysame, \emph{On fuzzy logic {III}: {S}emantical completeness of some
  many-valued propositional calculi}, Mathematical Logic Quarterly \textbf{25}
  (1979), no.~25--29, 447--464.

\bibitem{Po:FB}
Silke Pollandt, \emph{Fuzzy-{B}egriffe: {F}ormale {B}egriffsanalyse unscharfer
  {D}aten}, Springer, 1997.

\bibitem{PrTe:Gdratiuivq}
Henri Prade and Claudette Testemale, \emph{Generalizing database relational
  algebra for the treatment of incomplete or uncertain information and vague
  queries}, Information Sciences \textbf{34} (1984), no.~2, 115--143.

\bibitem{RaMa:Ffdljdfrds}
K.~V. S. V.~N. Raju and Arun~K. Majumdar, \emph{Fuzzy functional dependencies
  and lossless join decomposition of fuzzy relational database systems}, ACM
  Transactions on Database Systems \textbf{13} (1988), no.~2, 129--166.

\bibitem{RaSi}
Helena Rasiowa and Roman Sikorski, \emph{A proof of the completeness theorem of
  {G}\"odel}, Fundam. Math. \textbf{37} (1950), 193--200.

\bibitem{SaDePaFa:Ebrddfpl}
Yehoshua Sagiv, Claude Delobel, D.~Scott Parker, Jr., and Ronald Fagin,
  \emph{An equivalence between relational database dependencies and a fragment
  of propositional logic}, J. ACM \textbf{28} (1981), no.~3, 435--453.

\bibitem{DiWa}
Morgan Ward and Robert~P. Dilworth, \emph{Residuated lattices}, Trans. Amer.
  Math. Soc. \textbf{45} (1939), 335--354.

\bibitem{Wec:UAfCS}
Wolfgang Wechler, \emph{Universal {A}lgebra for {C}omputer {S}cientists}, EATCS
  Monographs on Theoretical Computer Science, vol.~25, Springer-Verlag, Berlin
  Heidelberg, 1992.

\bibitem{Za:Afstilh}
Lotfi~A. Zadeh, \emph{A fuzzy-set-theoretic interpretation of linguistic
  hedges}, Journal of Cybernetics \textbf{2} (1972), no.~3, 4--34.

\bibitem{Zak:Mnrar}
Mohammed~J. Zaki, \emph{Mining non-redundant association rules}, Data Mining
  and Knowledge Discovery \textbf{9} (2004), 223--248.

\end{thebibliography}

\end{document}